\documentclass[12pt]{article}
\newif\iflongversion
\longversiontrue

\newif\ifpredicationals
\predicationalsfalse

\usepackage{fullpage}
\usepackage{times}

\usepackage[bookmarks=false]{hyperref}

\usepackage{doi}

\usepackage{amsthm}

\theoremstyle{plain}
\newtheorem{theorem}{Theorem}
\newtheorem{lemma}[theorem]{Lemma}

\newtheorem{corollary}[theorem]{Corollary}

\theoremstyle{definition}
\newtheorem{definition}{Definition}
\theoremstyle{remark}
\newtheorem{remark}{Remark}

\usepackage{amsmath}
\usepackage{paralist}

\usepackage{math}
\usepackage[prefixflatinterpret,bracketmodalinterpret,fixformat,sidenotecalculus,longseqcontext,boldmquant,colormquant,seqarrow,seqoptional]{logic}
\usepackage[prefixflatinterpret,bracketmodalinterpret,fixformat,differentialdL]{dL}

\usepackage{tikz}
\usetikzlibrary{shapes}

\usepackage{xcolor}
\definecolor{vgreen}{rgb}{.1,.5,0}
\definecolor{vred}{rgb}{.7,0,0}
\definecolor{vblue}{rgb}{.1,.15,.62}

\tikzstyle{transition}=[-stealth]
\tikzstyle{similar state}=[double]
\tikzstyle{transition exists}=[dashed]
\tikzstyle{equivalence}=[<->,double]

\tikzstyle{strategy preimage}=[draw=vblue,fill=vblue!20,shape=ellipse]
\tikzstyle{winning condition}=[draw=vred,fill=vred!50,text=white,shape=circle]
\tikzstyle{winning annotation}=[ultra thick,black,->]

\def\qedhere{\qed}
\usepackage[keepproof]{endproof2}

\def\qedhere{}

\usepackage{prettyref}
\newcommand{\rref}[2][]{\prettyref{#2}}
\newrefformat{sec}{Section\,\ref{#1}}
\newrefformat{def}{Def.\,\ref{#1}}
\newrefformat{thm}{Theorem\,\ref{#1}}
\newrefformat{prop}{Proposition\,\ref{#1}}
\newrefformat{lem}{Lemma\,\ref{#1}}
\newrefformat{cor}{Corollary\,\ref{#1}}
\newrefformat{rem}{Rem.\,\ref{#1}}
\newrefformat{ex}{Example\,\ref{#1}}
\newrefformat{tab}{Table\,\ref{#1}}
\newrefformat{fig}{Fig.\,\ref{#1}}
\newrefformat{case}{case\,\ref{#1}}
\iflongversion
\newrefformat{app}{Appendix\,\ref{#1}}
\else
\newrefformat{app}{\cite{DBLP:journals/corr/Platzer18:usubst}}
\fi

\renewcommand*{\applyusubst}[2]{#1{#2}}%
\newcommand*{\usubstgroup}[1]{(#1)}

\renewcommand{\with}{\mathrel{:}}
\newcommand*{\ignore}[1]{}

\newcommand*{\linterpretationsconst}[2]{\mathcal{I}}
\renewcommand{\linterpretationsname}{\mathcal{S}}
\newcommand{\I}{\vdLint[const=I,state=\omega]}
\newcommand{\It}{\vdLint[const=I,state=\nu]}
\newcommand{\Iz}{\vdLint[const=I,state=\mu]}
\newcommand{\If}{\DALint[const=I,flow=\varphi]}
\newcommand*{\Iff}[1][\zeta]{\vdLint[const=I,state=\varphi(#1)]}%
\newcommand{\Idot}{\vdLint[const=I,state=]}

\newcommand{\Ia}{\iadjointSubst{\sigma}{\I}}%
\newcommand{\Ita}{\iadjointSubst{\sigma}{\It}}%
\newcommand{\Iminner}{\imodif[const]{\I}{\,\usarg}{d}}%

\newcommand{\Iat}{\iconcat[state=\nu]{\Ia}}%
\newcommand{\Itar}{\iconcat[state=\omega]{\Ita}}%

\newcommand{\IJalt}{\vdLint[const=J,state=\tilde{\omega}]}%
\newcommand{\IJtalt}{\vdLint[const=J,state=\tilde{\nu}]}%
\newcommand{\IIalt}{\vdLint[const=I,state=\tilde{\omega}]}%
\newcommand{\IItalt}{\vdLint[const=I,state=\tilde{\nu}]}%
\let\Ialt\IJalt
\let\Italt\IJtalt

\newcommand{\Imidi}[1][i]{\vdLint[const=I,state=\omega_{#1}]}

\newcommand{\Iaz}{\iconcat[state=\mu]{\Ia}}%
\newcommand{\Itaz}{\iconcat[state=\mu]{\Ita}}%

\newcommand{\bebecomes}{\mathrel{::=}}
\newcommand{\alternative}{~|~}
\providecommand{\dfn}[2][]{\emph{#2}}

\newcommand*{\iwinreg}[3][]{\iaccess[#1]{#2}\big(#3\big)}
\newcommand*{\iwin}[3][]{\iget[state]{#2} \in \iwinreg[#1]{#2}{#3}}
\newcommand*{\inowin}[3][]{\iget[state]{#2} \not\in \iwinreg[#1]{#2}{#3}}

\newcommand*{\restrictto}[2]{#1{\uparrow}#2}

\newcommand*{\selectlike}[3]{#1{\downarrow} #2{\scriptstyle(#3)}}
\newcommand*{\iselectlike}[3]{\selectlike{#1}{\iget[state]{#2}}{#3}}

\newcommand*{\genDE}[1]{\theta}%
\newcommand{\ivr}{\psi}

  \makeatletter
  \renewcommand{\iadjointSubst}[2]{%
    \useinterpretation{#2}%
    \edef\tmpadjointconst{{#1}^*_{\Interpretation@state}{\Interpretation@const}}%
    \iconcat[const=\tmpadjointconst]{#2}
  }
  \makeatother

\usepackage{ifpdf}
\ifpdf
\fi

\usepackage{fancyhdr}
\title{Uniform Substitution for Differential Game Logic}
\author{Andr\'e Platzer\thanks{
  Computer Science Department, Carnegie Mellon University, Pittsburgh, USA
  {aplatzer@cs.cmu.edu}
}}
\date{}

\begin{document}
\maketitle
\allowdisplaybreaks
\thispagestyle{empty}

\begin{abstract}
This paper presents a uniform substitution calculus for \emph{differential game logic} (\dGL).
Church's \emph{uniform substitutions} substitute a term or formula for a function or predicate symbol everywhere. 
After generalizing them to differential game logic and allowing for the substitution of hybrid games for game symbols, uniform substitutions make it possible to \emph{only} use axioms instead of axiom schemata, thereby substantially simplifying implementations.
Instead of subtle schema variables and soundness-critical side conditions on the occurrence patterns of logical variables to restrict infinitely many axiom schema instances to sound ones, the resulting axiomatization adopts only a finite number of ordinary \dGL formulas as axioms, which uniform substitutions instantiate soundly.
This paper proves soundness and completeness of uniform substitutions for the monotone modal logic \dGL.
The resulting axiomatization admits a straightforward modular implementation of \dGL in theorem provers.
\\[\medskipamount]
\end{abstract}

\newsavebox{\Rval}%
\sbox{\Rval}{$\scriptstyle\mathbb{R}$}

\newsavebox{\USarg}%
\sbox{\USarg}{$\boldsymbol{\cdot}$}

\newsavebox{\UScarg}%
\sbox{\UScarg}{$\boldsymbol{\_}$}

\section{Introduction}

Church's \emph{uniform substitution} is a classical proof rule for first-order logic \cite[\S35/40]{Church_1956}.
Uniform substitutions uniformly instantiate function and predicate symbols with terms and formulas, respectively, as functions of their arguments.
If $\phi$ is valid, then so is any admissible instance \(\applyusubst{\sigma}{\phi}\) for any uniform substitution $\sigma$:
\[
      \cinferenceRule[US|US]{uniform substitution}
      {\linferenceRule[formula]
        {\phi}
        {\applyusubst{\sigma}{\phi}}
      }{}%
\]
Uniform substitution
\(\sigma = \usubstlist{\usubstmod{p(\usarg)}{x+\usarg^2\geq\usarg}}\), e.g.\
turns \(\phi \mnodefequiv (p(4\itimes y) {\limply} \lexists{y}{p(x^2{+}y)})\) into
\(\applyusubst{\sigma}{\phi} \mequiv (x+(4\itimes y)^2\geq4\itimes y \limply \lexists{y}{x+(x^2+y)^2\geq x^2+y})\).
The introduction of $x$ is sound, but introducing variable $y$ via \(\sigma = \usubstlist{\usubstmod{p(\usarg)}{y+\usarg^2\geq\usarg}}\) would not be.
The occurrence of the variable $y$ of the argument $x^2+y$ that was already present previously, however, can correctly continue to be used in the instantiation.

\emph{Differential game logic} (\dGL), which is the specification and verification logic for \emph{hybrid games} \cite{DBLP:journals/tocl/Platzer15}, originally adopted uniform substitution for predicates, because they streamline and simplify completeness proofs.
A subsequent investigation of uniform substitutions for differential \emph{dynamic} logic (\dL) for hybrid \emph{systems} \cite{DBLP:journals/jar/Platzer17} confirmed how impressively Church's original motivation for uniform substitutions manifests in significantly simplifying prover implementations.

Church developed uniform substitutions to relate the study of (object-level) axioms to that of (meta-level) axiom schemata (which stand for an infinite family of axioms).
Beyond their philosophical considerations, uniform substitutions significantly impact prover designs by eliminating the usual gap between a logic and its prover.
After implementing the recursive application of uniform substitutions, the soundness-critical part of a theorem prover reduces to providing a copy of each concrete logical formula that the logic adopts as axioms.
Uniform substitutions provide a modular interface to the static semantics of the logic, because they are the only soundness-critical part of the prover that needs to know free or bound variables of an expression.
This simplicity is to be contrasted with the subtle soundness-critical side conditions that usually infest axiom schema and proof rule schema implementations, especially for the more involved binding structures of program logics.
The beneficial impact of uniform substitutions on provers made it possible to reduce the size of the soundness-critical core of the differential dynamic logic prover \KeYmaeraX \cite{DBLP:conf/cade/FultonMQVP15} down to 2\% compared to the previous prover \KeYmaera \cite{DBLP:conf/cade/PlatzerQ08} and formally verify \dL in Isabelle and Coq \cite{DBLP:conf/cpp/BohrerRVVP17}.

This paper generalizes uniform substitution to the significantly more expressive differential game logic for hybrid \emph{games} \cite{DBLP:journals/tocl/Platzer15}.
The modular structure of the soundness argument for \dL is sufficiently robust to work for \dGL:
\begin{inparaenum}[\it i)]
\item prove correctness of the static semantics,
\item relate syntactic effect of uniform substitution to semantic effect of its adjoint interpretation,
\item conclude soundness of rule \irref{US}, and
\item separately establish soundness of each axiom.
\end{inparaenum}
The biggest challenge is that hybrid game semantics cannot use state reachability, so correctness notions and their uses for the static semantics need to be phrased as functions of winning condition projections.
The interaction of game operators with repetitions causes transfinite fixpoints instead of the arbitrary finite iterations in hybrid systems. 
Relative completeness follows from previous results, but exploits the new game symbols to simplify the proof.
After new soundness justifications, the resulting uniform substitution mechanism and axioms for \dGL end up close to those for hybrid systems \cite{DBLP:journals/jar/Platzer17} (apart from the ones that are unsound for hybrid games \cite{DBLP:journals/tocl/Platzer15}).
The modularity caused by uniform substitutions explains why it was possible to generalize the \KeYmaeraX prover kernel from hybrid systems to hybrid games with about 10 lines of code.%
\footnote{%
The addition %
of games to the previous \KeYmaera prover was more complex \cite{DBLP:conf/cade/QueselP12}, with an implementation effort measured in months not minutes.
Unfortunately, this is not quite comparable, because both provers implement markedly different flavors of games for hybrid systems.
The game logic for \KeYmaera \cite{DBLP:conf/cade/QueselP12} was specifically tuned as an exterior extension to be more easily implementable than \dGL in \KeYmaera.
}
All proofs are inline.

\section{Preliminaries: Differential Game Logic}
This section reviews differential game logic (\dGL), a specification and verification logic for hybrid games \cite{DBLP:journals/tocl/Platzer15,DBLP:journals/tocl/Platzer17}.
Hybrid games support the discrete, continuous, and adversarial dynamics of two-player games in hybrid systems between players Angel and Demon.
Compared to previous work \cite{DBLP:journals/tocl/Platzer15}, the logic is augmented to form \emph{(differential-form) differential game logic} with differentials and function symbols \cite{DBLP:journals/jar/Platzer17} and with game symbols $a$ that can be substituted with hybrid games.

\subsection{Syntax}
Differential game logic has three syntactic categories.
Its terms $\theta$ are polynomial terms, function symbols interpreted over $\reals$, and differential terms $\der{\theta}$.
Its hybrid games $\alpha$ describe the permitted player actions during the game in program notation.
Its formulas $\phi$ include first-order logic of real arithmetic and, for each hybrid game $\alpha$, a modal formula \(\ddiamond{\alpha}{\phi}\), which expresses that player Angel has a winning strategy in the hybrid game $\alpha$ to reach the region satisfying \dGL formula $\phi$.
In the formula \(\ddiamond{\alpha}{\phi}\), the \dGL formula $\phi$ describes Angel's objective while the hybrid game $\alpha$ describes the moves permitted for the two players, respectively.

The set of all \emph{variables} is $\allvars$.
Variables of the form $\D{x}$ for a variable $x\in\allvars$ are called \emph{differential variables}, which are just independent variables associated to variable $x$. 
For any subset $V\subseteq\allvars$ is \(\D{V}\mdefeq\{\D{x} : x\in V\}\) the set of \emph{differential variables} $\D{x}$ for the variables in $V$.
The set of all variables is assumed to contain all its differential variables $\D{\allvars}\subseteq\allvars$ (although $\D[2]{x},\D[3]{x}$ are not usually used).

\begin{definition}[Terms]
\emph{Terms} are defined by this grammar
(with $\theta,\eta,\theta_1,\dots,\theta_k$ as terms, $x\in\allvars$ as variable, and $f$ as function symbol of arity $k$):
\[
  \theta,\eta ~\bebecomes~
  x
  \alternative
  f(\theta_1,\dots,\theta_k)
  \alternative
  \theta+\eta
  \alternative
  \theta\cdot\eta
  \alternative \der{\theta}
\]
\end{definition}

As in \dL \cite{DBLP:journals/jar/Platzer17}, \emph{differentials} \(\der{\theta}\) of terms $\theta$ are exploited for the purpose of axiomatically internalizing reasoning about differential equations.
The differential \(\der{\theta}\) describes how the value of $\theta$ changes locally depending on how the values of its variables $x$ change, i.e., as a function of the values of the corresponding differential variables $\D{x}$.
Differentials reduce reasoning about \emph{differential equations} to reasoning about \emph{equations of differentials} \cite{DBLP:journals/jar/Platzer17} with their single-state semantics.

\begin{definition}[Hybrid games] \label{def:dGL-HG}
The \emph{hybrid games of differential game logic {\dGL}} are defined by the following grammar (with $\alpha,\beta$ as hybrid games, $a$ as game symbol, $x$ as variable, $\theta$ as term, and $\ivr$ as \dGL formula):
\[
  \alpha,\beta ~\bebecomes~
  a\alternative
  \pupdate{\pumod{x}{\theta}}
  \alternative
  \pevolvein{\D{x}=\genDE{x}}{\ivr}
  \alternative
  \ptest{\ivr}
  \alternative
  \alpha\cup\beta
  \alternative
  \alpha;\beta
  \alternative
  \prepeat{\alpha}
  \alternative
  \pdual{\alpha}
\]
\end{definition}
Atomic games are the following.
\emph{Game symbols} $a$ are uninterpreted.
The \emph{discrete assignment game} \(\pupdate{\pumod{x}{\theta}}\) evaluates term $\theta$ and assigns it to variable $x$.
The \emph{continuous evolution game} \(\pevolvein{\D{x}=\genDE{x}}{\ivr}\) allows Angel to follow  differential equation \m{\D{x}=\genDE{x}} for any real duration during which the evolution domain constraint $\ivr$ is true
(\(\pevolve{\D{x}=\genDE{x}}\) stands for \(\pevolvein{\D{x}=\genDE{x}}{\ltrue}\)).
If $\ivr$ is not true in the current state, then no solution exists and Angel loses the game.
\emph{Test game} \(\ptest{\ivr}\) has no effect except that Angel loses the game prematurely unless $\ivr$ is true in the current state.

Compound games are the following.
The \emph{game of choice} \(\pchoice{\alpha}{\beta}\) allows Angel to choose whether she wants to play game $\alpha$ or, instead, play game $\beta$.
The \emph{sequential game} \(\alpha;\beta\) first plays $\alpha$ and then plays $\beta$ (unless a player lost prematurely during $\alpha$).
The \emph{repeated game} \(\prepeat{\alpha}\) allows Angel to decide how often to repeat game $\alpha$ by inspecting the state reached after the respective $\alpha$ game to decide whether she wants to play another round.
The \emph{dual game} \(\pdual{\alpha}\) makes the players switch sides: all of Angel's decisions are now Demon's and all of Demon's decisions are now Angel's.
Where Angel would have lost prematurely in $\alpha$ (for failing a test or evolution domain) now Demon does in $\pdual{\alpha}$, and vice versa.
This makes game play interactive but semantically quite rich \cite{DBLP:journals/tocl/Platzer15}.
All other operations are definable, e.g., the game where Demon chooses between $\alpha$ and $\beta$ as \(\pdual{(\pchoice{\pdual{\alpha}}{\pdual{\beta}})}\).

\begin{definition}[\dGL formulas] \label{def:dGL-formula}
The \emph{formulas of differential game logic {\dGL}} are defined by the following grammar (with $\phi,\psi$ as \dGL formulas, $p$ as predicate symbol of arity $k$, $\theta,\eta,\theta_i$ as terms, $x$ as variable, and $\alpha$ as hybrid game):
  \[
  \phi,\psi ~\bebecomes~
  \theta\geq\eta \alternative
  p(\theta_1,\dots,\theta_k) \alternative
  \lnot \phi \alternative
  \phi \land \psi \alternative
  \lexists{x}{\phi} \alternative 
  \ddiamond{\alpha}{\phi}
  \]
\end{definition}

\newcommand{\sndvel}{x{+}1}%

\noindent
The box modality \(\dbox{\alpha}{}\) in formula \(\dbox{\alpha}{\phi}\) describes that the player Demon has a winning strategy to achieve $\phi$ in hybrid game $\alpha$.
But \dGL satisfies the determinacy duality
\(\dbox{\alpha}{\phi} \lbisubjunct \lnot\ddiamond{\alpha}{\lnot\phi}\) \cite[Theorem 3.1]{DBLP:journals/tocl/Platzer15}, which we now take as its definition to simplify matters.
Other operators are definable as usual, e.g., \(\lforall{x}{\phi}\) as \(\lnot\lexists{x}{\lnot\phi}\).
{\renewcommand{\sndvel}{x^2{+}1}%
The following \dGL formula, for example, expresses that Angel has a winning strategy to follow the differential equation \(\pevolve{\D{x}=v}\) to a state where $x>0$ even after Demon chooses \(\pupdate{\pumod{v}{2}}\) or \(\pupdate{\pumod{v}{\sndvel}}\) first:
\(\ddiamond{\pdual{(\pchoice{\pupdate{\pumod{v}{2}}}{\pupdate{\pumod{v}{\sndvel}}})}; \pevolve{\D{x}=v}}{\,x>0}\).
}%

\subsection{Semantics}

While the syntax of \dGL is close to that of \dL (with the only change being the addition of the duality operator $\pdual{}$), its semantics is significantly more involved, because it needs to recursively support \emph{interactive} game play, instead of mere reachability.
Variables may have different values in different states of the game.
A \dfn{state} \(\iget[state]{\I}\) is a mapping from the set of all variables $\allvars$ to the reals $\reals$. 
Also, \(\iget[state]{\imodif[state]{\I}{x}{r}}\) is the state that agrees with state \(\iget[state]{\I}\) except for variable $x$ whose value is $r\in\reals$.
The set of all states is denoted $\linterpretations{\Sigma}{\allvars}$.
The set of all subsets of $\linterpretations{\Sigma}{\allvars}$ is denoted \(\powerset{\linterpretations{\Sigma}{\allvars}}\).

The semantics of function, predicate, and game symbols is independent from the state.
They are interpreted by an \dfn{interpretation} $\iget[const]{\I}$ that maps
each arity $k$ function symbol $f$ to a $k$-ary smooth function \(\iget[const]{\I}(f) : \reals^k\to\reals\),
and each arity $k$ predicate symbol $p$ to a $k$-ary relation \(\iget[const]{\I}(p) \subseteq \reals^k\).
The semantics of differential game logic in interpretation $\iget[const]{\I}$ defines, for each formula $\phi$, the set of all states \(\imodel{\I}{\phi}\), in which $\phi$ is true.
Since hybrid games appear in \dGL formulas and vice versa, the semantics \(\iwinreg[\alpha]{\I}{X}\) of hybrid game $\alpha$ in interpretation $\iget[const]{\I}$ is defined by simultaneous induction (\rref{def:HG-semantics}) as the set of all states from which Angel has a winning strategy in hybrid game $\alpha$ to achieve $X$.
The real value of term $\theta$ in state $\iget[state]{\I}$ for interpretation $\iget[const]{\I}$ is denoted \(\ivaluation{\I}{\theta}\) and defined as usual.\footnote{%
Even if not critical here, differentials have a differential-form semantics \cite{DBLP:journals/jar/Platzer17} as the sum of all partial derivatives by $x\in\allvars$ multiplied by the corresponding values of $\D{x}$:\\
\(
\ivaluation{\I}{\der{\theta}}
=
\sum_{x\in\allvars} \iget[state]{\I}(\D{x}) \Dp[x]{\ivaluation{\Idot}{\theta}}(\iget[state]{\I})
= \sum_{x\in\allvars} \iget[state]{\I}(\D{x}) \Dp[x]{\ivaluation{\I}{\theta}}
\)}
An interpretation $\iget[const]{\I}$ maps each game symbol $a$ to a function \(\iget[const]{\I}(a) : \powerset{\linterpretations{\Sigma}{\allvars}} \to \powerset{\linterpretations{\Sigma}{\allvars}}\), where \(\iget[const]{\I}(a)(X) \subseteq \linterpretations{\Sigma}{\allvars}\) are the states from which Angel has a winning strategy to achieve $X\subseteq\linterpretations{Sigma}{V}$.

\begin{definition}[\dGL semantics] \label{def:dGL-semantics}
The \emph{semantics of a \dGL formula} $\phi$ for each interpretation $\iget[const]{\I}$ with a corresponding set of states $\linterpretations{\Sigma}{\allvars}$ is the subset \m{\imodel{\I}{\phi}\subseteq\linterpretations{\Sigma}{\allvars}} of states in which $\phi$ is true.
It is defined inductively as follows
\begin{enumerate}
\item \(\imodel{\I}{\theta\geq\eta} = \{\iget[state]{\I} \in \linterpretations{\Sigma}{\allvars} \with \ivaluation{\I}{\theta}\geq\ivaluation{\I}{\eta}\}\)
\item \(\imodel{\I}{p(\theta_1,\dots,\theta_k)} = \{\iget[state]{\I} \in \linterpretations{\Sigma}{\allvars} \with (\ivaluation{\I}{\theta_1},\dots,\ivaluation{\I}{\theta_k})\in\iget[const]{\I}(p)\}\)
\item \(\imodel{\I}{\lnot\phi} = \scomplement{(\imodel{\I}{\phi})}\)
\(= \linterpretations{\Sigma}{\allvars} \setminus \imodel{\I}{\phi}\)
is the complement of \(\imodel{\I}{\phi}\)
\item \(\imodel{\I}{\phi\land\psi} = \imodel{\I}{\phi} \cap \imodel{\I}{\psi}\)
\item
{\def\Im{\imodif[state]{\I}{x}{r}}%
\(\imodel{\I}{\lexists{x}{\phi}} =  \{\iget[state]{\I} \in \linterpretations{\Sigma}{\allvars} \with \iget[state]{\Im} \in \imodel{\I}{\phi} ~\text{for some}~r\in\reals\}\)
}
\item \(\imodel{\I}{\ddiamond{\alpha}{\phi}} = \iwinreg[\alpha]{\I}{\imodel{\I}{\phi}}\)
\end{enumerate}
A \dGL formula $\phi$ is \emph{valid in $\iget[const]{\I}$}, written \m{\iget[const]{\I}\models{\phi}}, iff it is true in all states, i.e., \m{\imodel{\I}{\phi}=\linterpretations{\Sigma}{\allvars}}.
Formula $\phi$ is \emph{valid}, written \m{\entails\phi}, iff \m{\iget[const]{\I}\models{\phi}} for all interpretations $\iget[const]{\I}$.
\end{definition}

\begin{definition}[Semantics of hybrid games] \label{def:HG-semantics}
The \emph{semantics of a hybrid game} $\alpha$ for each interpretation $\iget[const]{\I}$ is a function \m{\iwinreg[\alpha]{\I}{{\cdot}}} that, for each set of Angel's winning states \m{X\subseteq\linterpretations{\Sigma}{\allvars}}, gives the \emph{winning region}, i.e., the set of states \m{\iwinreg[\alpha]{\I}{X} \subseteq \linterpretations{\Sigma}{\allvars}} from which Angel has a winning strategy to achieve $X$ in $\alpha$ (whatever strategy Demon chooses). It is defined inductively as follows
\begin{enumerate}
\item \(\iwinreg[a]{\I}{X} = \iget[const]{\I}(a)(X)\)
\item \(\iwinreg[\pupdate{\pumod{x}{\theta}}]{\I}{X} = \{\iget[state]{\I} \in \linterpretations{\Sigma}{\allvars} \with \modif{\iget[state]{\I}}{x}{\ivaluation{\I}{\theta}} \in X\}\)

\item \(\iwinreg[\pevolvein{\D{x}=\genDE{x}}{\ivr}]{\I}{X} = \{\iget[state]{\I} \in \linterpretations{\Sigma}{\allvars} \with 
      \iget[state]{\I}=\iget[state]{\Iff[0]}\) on $\scomplement{\{\D{x}\}}$ and \(\iget[state]{\Iff[r]}\in X\)
  for some function \m{\iget[flow]{\If}:[0,r]\to\linterpretations{\Sigma}{\allvars}} of some duration $r$
  satisfying \m{\imodels{\If}{\D{x}=\genDE{x}\land\ivr}}$\}$
  \\
  where \m{\imodels{\If}{\D{x}=\genDE{x}\land\ivr}}
  iff
  \(\imodels{\Iff[\zeta]}{\D{x}=\genDE{x}\land\ivr}\)
  and
  \(\iget[state]{\Iff[0]}=\iget[state]{\Iff[\zeta]}\) on $\scomplement{\{x,\D{x}\}}$ 
  for all \(0{\leq}\zeta{\leq}r\)
  and
  \(\D[t]{\iget[state]{\Iff[t]}(x)}(\zeta)\) exists and equals \(\iget[state]{\Iff[\zeta]}(\D{x})\) for all \(0{\leq}\zeta{\leq}r\) if \m{r{>}0}.

\item \(\iwinreg[\ptest{\ivr}]{\I}{X} = \imodel{\I}{\ivr}\cap X\)
\item \(\iwinreg[\pchoice{\alpha}{\beta}]{\I}{X} = \iwinreg[\alpha]{\I}{X}\cup\iwinreg[\beta]{\I}{X}\)
\item \(\iwinreg[\alpha;\beta]{\I}{X} = \iwinreg[\alpha]{\I}{\iwinreg[\beta]{\I}{X}}\)
\item \(\iwinreg[\prepeat{\alpha}]{\I}{X} = \capfold\{Z\subseteq\linterpretations{\Sigma}{\allvars} \with X\cup\iwinreg[\alpha]{\I}{Z}\subseteq Z\}\)

\item \(\iwinreg[\pdual{\alpha}]{\I}{X} = \scomplement{(\iwinreg[\alpha]{\I}{\scomplement{X}})}\)
\end{enumerate}
\end{definition}

\noindent
The semantics \(\iwinreg[\pevolvein{\D{x}=\genDE{x}}{\ivr}]{\I}{X}\) is the set of all states from which there is a solution of the differential equation \(\D{x}=\genDE{x}\) of some duration that reaches a state in $X$ without ever leaving the set of all states \(\imodel{\I}{\ivr}\) where evolution domain constraint $\ivr$ is true.
The initial value of $\D{x}$ in state $\iget[state]{\I}$ is ignored for that solution.
It is crucial that \(\iwinreg[\prepeat{\alpha}]{\I}{X}\) gives a least fixpoint semantics to repetition \cite{DBLP:journals/tocl/Platzer15}.

\begin{lemma}[{Monotonicity \cite[Lem.\,2.7]{DBLP:journals/tocl/Platzer15}}] \label{lem:monotone}%
  The semantics is \emph{monotone}, i.e., \m{\iwinreg[\alpha]{\I}{X}\subseteq\iwinreg[\alpha]{\I}{Y}} for all $X\subseteq Y$.
\end{lemma}

\section{Static Semantics}

The central bridge between a logic and its uniform substitutions is the definition of its static semantics via its free and bound variables.
The static semantics captures static variable relationships that are more tractable than the full nuances of the dynamic semantics.
It will be used in crucial ways to ensure that no variable is introduced free into a context within which it is bound during the uniform substitution application.
It is imperative for the soundness of uniform substitution that the static semantics be sound, so expressions only depend on their free variables and only their bound variables change during hybrid games.

The most tricky part for the soundness justification for \dGL is that the semantics of hybrid games is not a reachability relation, such that the usual semantic characterizations of free and bound variables from programs do not work for hybrid games.
Hybrid games have a more involved winning region semantics.

The first step is to define \emph{upward projections} \(\restrictto{X}{V}\) that increase the winning region $X\subseteq\linterpretations{\Sigma}{\allvars}$ from the variables $V\subseteq\allvars$ to all states that are ``on $V$ like $X$'', i.e., similar on $V$ to states in $X$ (and arbitrary on complement $\scomplement{V}$).
The \emph{downward projection} \(\iselectlike{X}{\I}{V}\) shrinks the winning region $X$ and selects the values of state $\iget[state]{\I}$ on variables $V\subseteq\allvars$ to keep just those states of $X$ that agree with $\iget[state]{\I}$ on $V$.
\begin{definition}%
  \label{def:projections}%
  The set \(\restrictto{X}{V} =
  \{ \iget[state]{\It} \in \linterpretations{\Sigma}{\allvars} \with \mexists{\iget[state]{\I}\in X}{\iget[state]{\I}=\iget[state]{\It} ~\text{on}~V}\} \supseteq X\)
  extends $X\subseteq\linterpretations{\Sigma}{\allvars}$ to the states that agree on $V\subseteq\allvars$ with some state in $X$ (written $\mexistsquantifier$).
  The set \(\iselectlike{X}{\I}{V} = \{ \iget[state]{\It}\in X \with \iget[state]{\I}=\iget[state]{\It} ~\text{on}~V\} \subseteq X\)
  selects state $\iget[state]{\I}$ on $V\subseteq\allvars$ in $X\subseteq\linterpretations{\Sigma}{\allvars}$.
\end{definition}

\begin{remark} \label{rem:projections}
  It is easy to check these properties of up and down projections:
  \begin{compactenum}
  \item \label{case:restrict-compose} Composition:
  \(\restrictto{\restrictto{X}{V}}{W} = \restrictto{X}{(V\cap W)}\)
  \item \label{case:restrict-amon} Antimonotone: 
  \(\restrictto{X}{W} \subseteq \restrictto{X}{V}\) for all \(W \supseteq V\)
  \item \label{case:restrict-extreme}
  \(\restrictto{X}{\emptyset} = \linterpretations{\Sigma}{\allvars}\) (unless $X=\emptyset$) and \(\restrictto{X}{\allvars} = X\), where $\allvars$ is the set of all variables
  \item \label{case:select-compose} Composition:
  \(\iselectlike{\iselectlike{X}{\I}{V}}{\I}{W} = \iselectlike{X}{\I}{V\cup W}\)
  \item \label{case:select-amon} Antimonotone: 
  \(\iselectlike{X}{\I}{W} \subseteq \iselectlike{X}{\I}{V}\) for all \(W \supseteq V\)
  \item \label{case:select-extreme}
  \(\iselectlike{X}{\I}{\emptyset} = X\) and \(\iselectlike{X}{\I}{\allvars} = X\cap\{\iget[state]{\I}\}\).
  Thus, \(\iget[state]{\I} \in \iselectlike{X}{\I}{V}\) for any $V$ iff \(\iget[state]{\I} \in X\).
  \end{compactenum}
\end{remark}
\begin{proofatend}
\begin{compactenum}
\item[\ref{case:restrict-compose}.]
\(\restrictto{\restrictto{X}{V}}{W}\) are all states in $\linterpretations{\Sigma}{\allvars}$ that agree on $W$ with a state in \(\restrictto{X}{V}\), which, in turn, are all states that agree on $V$ with a state in $X$.
That is, \(\restrictto{\restrictto{X}{V}}{W}\) are all states that agree on $W$ with some state that agrees on $V$ with a state in $X$, which is the set \(\restrictto{X}{(V\cap W)}\) of states that agree on $V\cap W$ with a state in $X$.

\item[\ref{case:restrict-amon}.]
\(W\supseteq V\) implies \(V=W\cap U\) for some $U$.
By \rref{case:restrict-compose}, \(\restrictto{X}{V}=\restrictto{\restrictto{X}{W}}{U} \supseteq \restrictto{X}{W}\) by \rref{def:projections}.

\item[\ref{case:restrict-extreme}.]
First note \(\restrictto{\emptyset}{V}=\emptyset\) for all $V$.
If $X\neq\emptyset$, then \(\restrictto{X}{\emptyset} = \linterpretations{\Sigma}{\allvars}\), because equality on $\emptyset$ imposes no conditions on the state $\iget[state]{\It}$.
\(\restrictto{X}{\allvars} = X\), because agreement on all variables $\allvars$ implies $\iget[state]{\I}=\iget[state]{\It}$.

\item[\ref{case:select-compose}.]
\(\iselectlike{\iselectlike{X}{\I}{V}}{\I}{W}\) are all states that agree on $W$ with $\iget[state]{\I}$ and are in the set \(\iselectlike{X}{\I}{V}\).
That is, \(\iselectlike{\iselectlike{X}{\I}{V}}{\I}{W}\) are all states in $X$ that agree on $W$ and on $V$ with $\iget[state]{\I}$, which is the set \(\iselectlike{X}{\I}{V\cup W}\).

\item[\ref{case:select-amon}.]
\(W\supseteq V\) implies \(W=V\cup U\) for some $U$.
By \rref{case:select-compose}, \(\iselectlike{X}{\I}{W} = \iselectlike{\iselectlike{X}{\I}{V}}{\I}{U} \subseteq \iselectlike{X}{\I}{V}\) by \rref{def:projections}.

\item[\ref{case:select-extreme}.]
\(\iselectlike{X}{\I}{\emptyset} = X\) since agreement on $\emptyset$ imposes no conditions on $\iget[state]{\It}\in X$.
Furthermore, \(\iselectlike{X}{\I}{\allvars} = X\cap\{\iget[state]{\I}\}\) since agreement on all variables $\allvars$ imposes the condition $\iget[state]{\It}=\iget[state]{\I}$, which is in \(\iselectlike{X}{\I}{\allvars}\) iff $\iget[state]{\I}\in X$.
\qedhere
\end{compactenum}
\end{proofatend}

Projections make it possible to define (\emph{semantic!}) free and bound variables of hybrid games by expressing suitable variable dependence and ignorance.
Variable $x$ is free iff two states that only differ in the value of $x$ have different membership in the winning region for hybrid game $\alpha$ for some winning region $\restrictto{X}{\scomplement{\{x\}}}$ that is insensitive to the value of $x$.
Variable $x$ is bound iff it is in the winning region for hybrid game $\alpha$ for some winning condition $X$ but not for the winning condition \(\iselectlike{X}{\I}{\{x\}}\) that limits the new value of $x$ to stay at its initial value $\iget[state]{\I}(x)$.

\begin{definition}[Static semantics] \label{def:static-semantics}
The \emph{static semantics} defines the \emph{free variables}, which are all variables that the value of an expression depends on,
as well as \emph{bound variables}, $\boundvarsdef{\alpha}$, which can change their value during game $\alpha$, as:
\allowdisplaybreaks%
\begin{align*}
  \freevarsdef{\theta} &= 
  \big\{x \in \allvars \with \mexists{\iget[const]{\I},\iget[state]{\I},\iget[state]{\IIalt}} \text{such that}~ \iget[state]{\I}=\iget[state]{\IIalt} ~\text{on}~\scomplement{\{x\}} ~ \text{and} ~ \ivaluation{\I}{\theta}\neq\ivaluation{\IIalt}{\theta}\big\}
  \\
  \freevarsdef{\phi} &= \big\{x \in \allvars \with \mexists{\iget[const]{\I},\iget[state]{\I},\iget[state]{\IIalt}} \text{such that}~\iget[state]{\I}=\iget[state]{\IIalt} ~\text{on}~\scomplement{\{x\}} ~ \text{and} ~ \imodels{\I}{\phi}\not\ni\iget[state]{\IIalt}\big\}
  \\
  \freevarsdef{\alpha} &= \big\{x \in \allvars \with \mexists{\iget[const]{\I},\iget[state]{\I},\iget[state]{\IIalt},X}
  \text{with}~\iget[state]{\I}=\iget[state]{\IIalt} ~\text{on}~\scomplement{\{x\}} ~\text{and}~
  \iget[state]{\I} \in \iwinreg[\alpha]{\I}{\restrictto{X}{\scomplement{\{x\}}}} \not\ni \iget[state]{\IIalt}\big\}
  \\
  \boundvarsdef{\alpha} &= \big\{x \in \allvars \with \mexists{\iget[const]{\I},\iget[state]{\I},X} \text{such that} ~ \iwinreg[\alpha]{\I}{X} \ni \iget[state]{\I} \not\in \iwinreg[\alpha]{\I}{\iselectlike{X}{\I}{\{x\}}} \big\}
\end{align*}
The \dfn{signature}, i.e., set of function, predicate, 
%
and game symbols in $\phi$ is denoted \(\intsigns{\phi}\); accordingly \(\intsigns{\theta}\) for term $\theta$ and \(\intsigns{\alpha}\) for hybrid game $\alpha$.
\end{definition}

The static semantics from \rref{def:static-semantics} satisfies the coincidence property (the value of an expression only depends on the values of its free variables) and bound effect property (a hybrid game only changes the values of its bound variables).

\begin{lemma}[Coincidence for terms] \label{lem:coincidence-term}
  $\freevarsdef{\theta}$ is the smallest set with the coincidence property for $\theta$:
  If \(\iget[state]{\I}=\iget[state]{\Ialt}\) on $\freevarsdef{\theta}$
  and \(\iget[const]{\I}=\iget[const]{\Ialt}\) on $\intsigns{\theta}$ then
  \m{\ivaluation{\I}{\theta}=\ivaluation{\Ialt}{\theta}}.
\end{lemma}
\begin{proofatend}
By \cite[Lem.\,10]{DBLP:journals/jar/Platzer17}, as semantics and free variables of terms are as in \dL. \qedhere
\end{proofatend}

\begin{lemma}[Coincidence for formulas] \label{lem:coincidence}
  $\freevarsdef{\phi}$ is the smallest set with the coincidence property for $\phi$:
  If \(\iget[state]{\I}=\iget[state]{\Ialt}\) on $\freevarsdef{\phi}$
  and \(\iget[const]{\I}=\iget[const]{\Ialt}\) on $\intsigns{\phi}$, then
  \m{\imodels{\I}{\phi}} iff \m{\imodels{\Ialt}{\phi}}.
\end{lemma}
\begin{proofatend}
The semantics of formulas and their semantic free variables is analogous to \dL, so \cite[Lem.\,11]{DBLP:journals/jar/Platzer17} transfers, because its proof is by induction on the set of free variables independently of the particular syntactic structure of the formula $\phi$ and, thus, the proof is not affected by the modified meaning of modalities.
\qedhere
\end{proofatend}

\noindent
From which states a hybrid game $\alpha$ can be won only depends on $\alpha$, the winning region, and the values of its free variables, as \(\restrictto{X}{\freevarsdef{\alpha}}\) is only sensitive to $\freevarsdef{\alpha}$.

\begin{lemma}[Coincidence for games] \label{lem:coincidence-HG}
  The set $\freevarsdef{\alpha}$ is the smallest set with the coincidence property for $\alpha$:
  If \(\iget[state]{\I}=\iget[state]{\Ialt}\) on $V\supseteq\freevarsdef{\alpha}$
  and \(\iget[const]{\I}=\iget[const]{\Ialt}\) on $\intsigns{\alpha}$,
  then 
  \(\iwin[\alpha]{\I}{\restrictto{X}{V}}\) iff
  \(\iwin[\alpha]{\Ialt}{\restrictto{X}{V}}\).

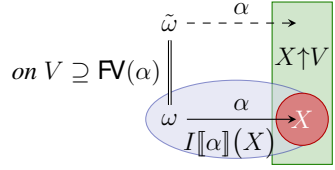
\centerline{
\begin{tikzpicture}[scale=0.8, every node/.style={transform shape}]
  \useasboundingbox (-3.2,-1.6) rectangle (2,1.2); 
    \node[draw=vgreen,fill=vgreen!20,shape=rectangle,minimum height=2.7cm,minimum width=1cm] (Xalt) at (1.2,-0.2) {};
    \node[draw=vblue,fill=vblue!20,shape=ellipse,minimum height=1.3cm,minimum width=3cm,opacity=0.5] at (0.1,-0.8) {};
    \node at (1.2,0.2) {\m{\restrictto{X}{V}}};
    \node[winning condition,minimum width=0.5cm] (X) at (1.2,-0.8) {$X$};
    \node at (0,-1.2) {\(\iwinreg[\alpha]{\I}{X}\)};
    \node (v) at (-1,-0.8) {$\iget[state]{\I}$};
    \node (valt) at (-1,0.8) {$\iget[state]{\Ialt}$};
  \path
    (v) edge [similar state] node [left,align=right] {on $V\supseteq\freevarsdef{\alpha}$} (valt)
            edge [transition] node [above] {$\alpha$} (1.1,-0.8)
    (valt) edge [transition,transition exists] node [above] {$\alpha$}
            (1.1,0.8);
\end{tikzpicture}
}
\end{lemma}
\begin{proof}
\let\Ialt\IIalt
\let\Italt\IItalt
Let $\mathcal{M}$ be the set of all sets $M\subseteq\allvars$ satisfying for all $\iget[const]{\I},\iget[state]{\I},\iget[state]{\Ialt},X$ that \(\iget[state]{\I}=\iget[state]{\Ialt}\) on $\scomplement{M}$ implies:
\(\iwin[\alpha]{\I}{\restrictto{X}{V}}\) iff 
\(\iwin[\alpha]{\Ialt}{V}\).
One implication suffices.
\begin{compactenum}
\item If $x\not\in V$, then \(\{x\}\in\mathcal{M}\):
Assume \(\iget[state]{\I}=\iget[state]{\Ialt}\) on $\scomplement{\{x\}}$ and
\(\iwin[\alpha]{\I}{\restrictto{X}{V}} \subseteq \iwinreg[\alpha]{\I}{\restrictto{\restrictto{X}{V}}{\{x\}}}\) by Lem.\,\ref{lem:monotone}, \rref{def:projections}.
Then, as $x\not\in\freevarsdef{\alpha}$,
\(\iwin[\alpha]{\Ialt}{\restrictto{\restrictto{X}{V}}{\{x\}}}\)
= \(\iwinreg[\alpha]{\Ialt}{\restrictto{X}{(V\cap\scomplement{\{x\}})}}\) by \rref{rem:projections}(\ref{case:restrict-compose}).
Finally, \(\restrictto{X}{(V\cap\scomplement{\{x\}})} = \restrictto{X}{V}\) as $x\not\in V$.

\item If $M_i\in\mathcal{M}$ is a sequence of sets in $\mathcal{M}$, then \(\cupfold_{i\in\naturals} M_i \in \mathcal{M}\):
Assume \(\iget[state]{\I}=\iget[state]{\Ialt}\) on \(\scomplement{(\cupfold_i M_i)}\)
and \(\iwin[\alpha]{\I}{\restrictto{X}{V}}\).
The state $\iget[state]{\Imidi[n]}$ defined as $\iget[state]{\Ialt}$ on \(\cupfold_{i<n} M_i\) and as $\iget[state]{\I}$ on \(\scomplement{(\cupfold_{i<n} M_i)}\) satisfies \(\iwin[\alpha]{\Imidi[n]}{\restrictto{X}{V}}\) by induction on $n$.
For $n=0$, $\iget[state]{\Imidi[0]}=\iget[state]{\I}$.
Since \(\iget[state]{\Imidi[n]}=\iget[state]{\Imidi[n+1]}\) on $\scomplement{M_n}$ and $M_n\in\mathcal{M}$, 
\(\iwin[\alpha]{\Imidi[n]}{\restrictto{X}{V}}\)
implies
\(\iwin[\alpha]{\Imidi[n+1]}{\restrictto{X}{V}}\).
Finally, \(\iget[state]{\I}=\iget[state]{\Ialt}=\iget[state]{\Imidi[n]}\) on \(\scomplement{(\cupfold_i M_i)}\) already.
\end{compactenum}
This argument succeeds for any $V\supseteq\freevarsdef{\alpha}$, so \(\scomplement{\freevarsdef{\alpha}} \in \mathcal{M}\) as a (countable) union of \(\{x\}\) for all $x\not\in\freevarsdef{\alpha}$.
{%
\let\Ialt\IJalt%
\let\Italt\IJtalt%
Finally, if \(\iget[const]{\I}=\iget[const]{\Ialt}\) on $\intsigns{\alpha}$ then also \(\iwin[\alpha]{\Ialt}{\restrictto{X}{V}}\) by a simple induction, since $\iget[const]{\I}$ gives meaning to function, predicate,
%
and game symbols, but only those that occur in $\alpha$ are relevant.
}

No set $W\not\supseteq\freevarsdef{\alpha}$ has the coincidence property for $\alpha$,
because there, then, is a variable $x\in\freevarsdef{\alpha}\setminus W$,
which implies there are \(\iget[const]{\I},X, \iget[state]{\I}=\iget[state]{\Ialt}\) on $\scomplement{\{x\}}\supseteq W$ such that
\(\iget[state]{\I} \in \iwinreg[\alpha]{\I}{\restrictto{X}{\scomplement{\{x\}}}} \not\ni \iget[state]{\Ialt}\).
But for the set $V\mdefeq\scomplement{\{x\}} \supseteq W$ it is, then, the case that 
\(\iwin[\alpha]{\I}{\restrictto{X}{V}}\) but
\(\inowin[\alpha]{\Ialt}{\restrictto{X}{V}}\).
\qedhere
\end{proof}

\noindent
By \rref{def:projections} and \rref{lem:monotone}, \(\iwin[\alpha]{\I}{X}\) implies \(\iwin[\alpha]{\I}{\restrictto{X}{V}}\) for all $V\subseteq\allvars$.
All supersets of \(\freevarsdef{\theta}\) or \(\freevarsdef{\phi}\) or \(\freevarsdef{\alpha}\) have the respective coincidence property.

Only its bound variables $\boundvarsdef{\alpha}$ change their values during hybrid game $\alpha$, because from any state from which $\alpha$ can be won to achieve $X$, one can already win $\alpha$ to achieve $\iselectlike{X}{\I}{\scomplement{\boundvarsdef{\alpha}}}$, which stays at $\iget[state]{\I}$ except for the values of $\boundvarsdef{\alpha}$.

\begin{lemma}[Bound effect] \label{lem:bound}
  The set $\boundvarsdef{\alpha}$ is the smallest set with the bound effect property: 
  \(\iwin[\alpha]{\I}{X}\) iff \(\iwin[\alpha]{\I}{\iselectlike{X}{\I}{\scomplement{\boundvarsdef{\alpha}}}}\).

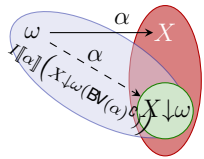
\centerline{
\begin{tikzpicture}[scale=0.8, every node/.style={transform shape}]
  \useasboundingbox (-3.2,-1.3) rectangle (2,1.3); 
    \node[winning condition,shape=ellipse,minimum height=2.5cm,minimum width=1.2cm] (X) at (1.2,0) {};
    \node[draw=vblue,fill=vblue!20,shape=ellipse,minimum height=1.5cm,minimum width=3.3cm,opacity=0.5,rotate=-30] at (0.1,0.1) {};
    \node[winning condition,draw=none,fill=none] at (1.2,0.8) {$X$};
    \node[winning condition,draw=vgreen,fill=vgreen!20,text=black,minimum width=0.5cm,inner sep=0pt] (selectX) at (1.2,-0.5) {$X{\downarrow}\iget[state]{\I}$};
    \node[rotate=-30] at (-0.2,-0.1) {\(\scriptstyle\iwinreg[\alpha]{\I}{\iselectlike{X}{\I}{\scomplement{\boundvarsdef{\alpha}}}}\)};
    \node (v) at (-1,0.8) {$\iget[state]{\I}$};
  \path
    (v) edge [transition] node [above,pos=0.7] {$\alpha$} (1,0.8);
  \path
    (v) edge [transition,dashed] node [above] {$\alpha$} (selectX);
\end{tikzpicture}
}
\end{lemma}
\begin{proofatend}
Let $\mathcal{M}$ be the set of all sets $M\subseteq\allvars$ satisfying for all $\iget[const]{\I},X,\iget[state]{\I}$: 
\(\iwin[\alpha]{\I}{X}\) iff \(\iwin[\alpha]{\I}{\iselectlike{X}{\I}{M}}\).
By  \rref{lem:monotone},
\(\iwinreg[\alpha]{\I}{X} \supseteq \iwinreg[\alpha]{\I}{\iselectlike{X}{\I}{M}}\) as \(X \supseteq \iselectlike{X}{\I}{M}\).

\begin{compactenum}
\item If \(x\not\in\boundvarsdef{\alpha}\), then \(\{x\} \in \mathcal{M}\) directly by \rref{def:static-semantics}.
\item If \(M_i \in\mathcal{M}\) is a sequence of sets in $\mathcal{M}$, then \(\cupfold_{i\in\naturals} M_i \in \mathcal{M}\):
Assume that\\
\m{\iwin[\alpha]{\I}{X} = \iwinreg[\alpha]{\I}{\iselectlike{X}{\I}{\emptyset}}} by \rref{rem:projections}(\ref{case:select-extreme}).
Since
\(\iwin[\alpha]{\I}{\iselectlike{X}{\I}{\cupfold_{i<n} M_i}}\)
implies
\(\iwin[\alpha]{\I}{\iselectlike{\iselectlike{X}{\I}{\cupfold_{i<n} M_i}}{\I}{M_n}}\)
= \(\iwinreg[\alpha]{\I}{\iselectlike{X}{\I}{\cupfold_{i<n+1} M_i}}\) according to \rref{rem:projections}(\ref{case:select-compose}),
an induction on $n$ yields
\(\iwin[\alpha]{\I}{\iselectlike{X}{\I}{\cupfold_i M_i}}\).
\end{compactenum}
Thus, \(\scomplement{\boundvarsdef{\alpha}} \in \mathcal{M}\) as a (countable) union of \(\{x\}\) for all $x\not\in\boundvarsdef{\alpha}$.

No set \(V\not\supseteq\boundvarsdef{\alpha}\) has the bound effect property for $\alpha$, because there, then, is a variable $x\in\boundvarsdef{\alpha}\setminus V$, which implies there are \(\iget[const]{\I},X,\iget[state]{\I}\) such that
\(\iwinreg[\alpha]{\I}{X} \ni \iget[state]{\I} \not\in \iwinreg[\alpha]{\I}{\iselectlike{X}{\I}{\{x\}}}
\supseteq \iwinreg[\alpha]{\I}{\iselectlike{X}{\I}{\scomplement{V}}}\)
by \rref{lem:monotone},
as \(\iselectlike{X}{\I}{\{x\}} \supseteq \iselectlike{X}{\I}{\scomplement{V}}\) by \rref{rem:projections}(\ref{case:select-amon}),
because \(\{x\} \subseteq \scomplement{V}\).
\qedhere
\end{proofatend}

\noindent
All supersets \(V\supseteq\boundvarsdef{\alpha}\) have the bound effect property, as
\(\iwinreg[\alpha]{\I}{\iselectlike{X}{\I}{\scomplement{V}}} \supseteq \iwinreg[\alpha]{\I}{\iselectlike{X}{\I}{\scomplement{\boundvarsdef{\alpha}}}}\)
by \rref{rem:projections}(\ref{case:select-amon}) 
because \(\scomplement{V} \subseteq \scomplement{\boundvarsdef{\alpha}}\).
Other states that agree except on the bound variables share the same selection of the winning region:
if \(\iget[state]{\I}=\iget[state]{\Ialt}\) on $\scomplement{\boundvarsdef{\alpha}}$, then
  \(\iget[state]{\Ialt} \in \iwinreg[\alpha]{\I}{X}\) iff \(\iget[state]{\Ialt} \in \iwinreg[\alpha]{\I}{\iselectlike{X}{\I}{\scomplement{\boundvarsdef{\alpha}}}}\).

Since all supersets of the free variables have the coincidence property and all supersets of the bound variables have the bound effect property, algorithms that \emph{syntactically compute} supersets \text{FV} and \text{BV} of free and bound variables \cite[Lem.\,17]{DBLP:journals/jar/Platzer17} can be soundly augmented by
\(\freevars{\pdual{\alpha}} = \freevars{\alpha}\)
and
\(\boundvars{\pdual{\alpha}} = \boundvars{\alpha}\).

\section{Uniform Substitution}

The static semantics provides, in a modular way, what is needed to define the application \(\applyusubst{\sigma}{\phi}\) of uniform substitution $\sigma$ to \dGL formula $\phi$.
The \dGL axiomatization uses uniform substitutions that affect terms, formulas, and games, whose application \(\applyusubst{\sigma}{\phi}\) will be defined in \rref{def:usubst-admissible} using \rref{fig:usubst}.
A \dfn{uniform substitution} $\sigma$ is a mapping
from expressions of the
form \(f(\usarg)\) to terms $\applysubst{\sigma}{f(\usarg)}$,
from \(p(\usarg)\) to formulas $\applysubst{\sigma}{p(\usarg)}$,
%
and from game symbols \(a\) to hybrid games $\applysubst{\sigma}{a}$.
Vectorial extensions are accordingly for other arities $k\geq0$.
Here $\usarg$ is a reserved function symbol of arity 0,
%
marking the position where the respective argument, e.g., argument $\theta$ to $p(\usarg)$ in formula $p(\theta)$, will end up in the replacement $\applysubst{\sigma}{p(\usarg)}$ used for $p(\theta)$.

\begin{figure}[tb]
  \newcommand*{\implied}[1]{\textcolor{darkgray}{#1}}%
  \renewcommand{\mequiv}{=}%
  \begin{displaymath}
    \begin{array}{@{}rcll@{}}
    \applyusubst{\sigma}{\usubstgroup{x}} &=& x & \text{for variable $x\in\allvars$}\\
    \applyusubst{\sigma}{\usubstgroup{f(\theta)}} &=& (\applyusubst{\sigma}{f})(\applyusubst{\sigma}{\theta})
  \mdefeq \applyusubst{\{\usarg\mapsto\applyusubst{\sigma}{\theta}\}}{\applysubst{\sigma}{f(\usarg)}} &
  \text{for function symbol}~f%
  \\
  \applyusubst{\sigma}{\usubstgroup{\theta+\eta}} &=& \applyusubst{\sigma}{\theta} + \applyusubst{\sigma}{\eta}
  \\
  \applyusubst{\sigma}{\usubstgroup{\theta\cdot\eta}} &=& \applyusubst{\sigma}{\theta} \cdot \applyusubst{\sigma}{\eta}
  \\
  \applyusubst{\sigma}{\usubstgroup{\der{\theta}}} &=& \der{\applyusubst{\sigma}{\theta}} &\text{if $\sigma$ is $\allvars$-admissible for $\theta$}
  \\
  \hline
  \applyusubst{\sigma}{\usubstgroup{\theta\geq\eta}} &\mequiv& \applyusubst{\sigma}{\theta} \geq \applyusubst{\sigma}{\eta}\\
    \applyusubst{\sigma}{\usubstgroup{p(\theta)}} &\mequiv& (\applyusubst{\sigma}{p})(\applyusubst{\sigma}{\theta})
  \mdefequiv \applyusubst{\{\usarg\mapsto\applyusubst{\sigma}{\theta}\}}{\applysubst{\sigma}{p(\usarg)}} &
  \text{for predicate symbol}~p%
  \\
    \applyusubst{\sigma}{\usubstgroup{\lnot\phi}} &\mequiv& \lnot\applyusubst{\sigma}{\phi}\\
    \applyusubst{\sigma}{\usubstgroup{\phi\land\psi}} &\mequiv& \applyusubst{\sigma}{\phi} \land \applyusubst{\sigma}{\psi}\\
    \applyusubst{\sigma}{\usubstgroup{\lexists{x}{\phi}}} &\mequiv& \lexists{x}{\applyusubst{\sigma}{\phi}} & \text{if $\sigma$ is $\{x\}$-admissible for $\phi$}\\
    \applyusubst{\sigma}{\usubstgroup{\ddiamond{\alpha}{\phi}}} &\mequiv& \ddiamond{\applyusubst{\sigma}{\alpha}}{\applyusubst{\sigma}{\phi}} & \text{if $\sigma$ is $\boundvarsdef{\applyusubst{\sigma}{\alpha}}$-admissible for $\phi$}
    \\
  \hline
    \applyusubst{\sigma}{\usubstgroup{a}} &\mequiv& \applysubst{\sigma}{a} &\text{for game symbol}~ a%
    \\
    \applyusubst{\sigma}{\usubstgroup{\pupdate{\umod{x}{\theta}}}} &\mequiv& \pupdate{\umod{x}{\applyusubst{\sigma}{\theta}}}\\
    \applyusubst{\sigma}{\usubstgroup{\pevolvein{\D{x}=\genDE{x}}{\ivr}}} &\mequiv&
    (\pevolvein{\D{x}=\applyusubst{\sigma}{\genDE{x}}}{\applyusubst{\sigma}{\ivr}}) & \text{if $\sigma$ is $\{x,\D{x}\}$-admissible for $\genDE{x},\ivr$}\\
    \applyusubst{\sigma}{\usubstgroup{\ptest{\ivr}}} &\mequiv& \ptest{\applyusubst{\sigma}{\ivr}}\\
    \applyusubst{\sigma}{\usubstgroup{\pchoice{\alpha}{\beta}}} &\mequiv& \pchoice{\applyusubst{\sigma}{\alpha}} {\applyusubst{\sigma}{\beta}}\\
    \applyusubst{\sigma}{\usubstgroup{\alpha;\beta}} &\mequiv& \applyusubst{\sigma}{\alpha}; \applyusubst{\sigma}{\beta} &\text{if $\sigma$ is $\boundvarsdef{\applyusubst{\sigma}{\alpha}}$-admissible for $\beta$}\\
    \applyusubst{\sigma}{\usubstgroup{\prepeat{\alpha}}} &\mequiv& \prepeat{(\applyusubst{\sigma}{\alpha})} &\text{if $\sigma$ is $\boundvarsdef{\applyusubst{\sigma}{\alpha}}$-admissible for $\alpha$}\\
    \applyusubst{\sigma}{\usubstgroup{\pdual{\alpha}}} &\mequiv& \pdual{(\applyusubst{\sigma}{\alpha})} 
    \end{array}%
  \end{displaymath}%
  \caption{Recursive application of uniform substitution~$\sigma$}%
  \index{substitution!uniform|textbf}%
  \label{fig:usubst}%
\end{figure}%

\begin{definition}[Admissible uniform substitution] \label{def:usubst-admissible}
  \index{admissible|see{substitution, admissible}}
  A uniform substitution~$\sigma$ is \emph{$U$-admissible for $\phi$} (or $\theta$ or $\alpha$, respectively) with respect to the variables $U\subseteq\allvars$ iff
  \(\freevarsdef{\restrict{\sigma}{\intsigns{\phi}}}\cap U=\emptyset\),
  where \({\restrict{\sigma}{\intsigns{\phi}}}\) is the \emph{restriction} of $\sigma$ that only replaces symbols that occur in $\phi$,
  and
  \(\freevarsdef{\sigma}=\cupfold_{f\ignore{\in\replacees{\sigma}}} \freevarsdef{\applysubst{\sigma}{f(\usarg)}} \cup \cupfold_{p\ignore{\in\replacees{\sigma}}} \freevarsdef{\applysubst{\sigma}{p(\usarg)}}\)
  are the \emph{free variables} that $\sigma$ introduces. 
  A uniform substitution~$\sigma$ is \emph{admissible for $\phi$} ($\theta$ or $\alpha$, respectively) iff the bound variables $U$ of each operator of $\phi$ are not free in the substitution on its arguments, i.e., $\sigma$ is $U$-admissible.
  These admissibility conditions are listed in \rref{fig:usubst}, which defines the result $\applyusubst{\sigma}{\phi}$ of applying $\sigma$ to $\phi$.
\end{definition}

\noindent
The remainder of this section proves soundness of uniform substitution for \dGL.
\emph{All subsequent uses of uniform substitutions are required to be admissible.}

\subsection{Uniform Substitution Lemmas}

Uniform substitution lemmas equate the syntactic effect that a uniform substitution $\sigma$ has on a syntactic expression in a state $\iget[state]{\I}$ and interpretation $\iget[const]{\I}$ with the semantic effect that the switch to the adjoint interpretation $\iget[const]{\Ia}$ has on the original expression.
Adjoints make it possible to capture in semantics the effect that a uniform substitution has on the syntax.

Let \m{\iget[const]{\imodif[const]{\I}{\,\usarg}{d}}} denote the interpretation that agrees with interpretation~$\iget[const]{\I}$ except for the interpretation of arity 0 function symbol~$\usarg$ which is changed to~\m{d\in\reals}.

\begin{definition}[Substitution adjoints] \label{def:adjointUsubst}
\def\Ialta{\iadjointSubst{\sigma}{\Ialt}}%
The \emph{adjoint} to substitution $\sigma$ is the operation that maps $\iportray{\I}$ to the \emph{adjoint} interpretation $\iget[const]{\Ia}$ in which the interpretation of each function symbol $f\ignore{\in\replacees{\sigma}}$, predicate symbol $p\ignore{\in\replacees{\sigma}}$,
%
and game symbol $a\ignore{\in\replacees{\sigma}}$ are modified according to $\sigma$ (it is enough to consider those that $\sigma$ changes):
\begin{align*}
  \iget[const]{\Ia}(f) &: \reals\to\reals;\, d\mapsto\ivaluation{\imodif[const]{\I}{\,\usarg}{d}}{\applysubst{\sigma}{f}(\usarg)}\\
  \iget[const]{\Ia}(p) &= \{d\in\reals \with \imodels{\imodif[const]{\I}{\,\usarg}{d}}{\applysubst{\sigma}{p}(\usarg)}\}
  \\
  \iget[const]{\Ia}(a) &: \powerset{\linterpretations{\Sigma}{\allvars}}\to\powerset{\linterpretations{\Sigma}{\allvars}};\, X\mapsto\iwinreg[\applysubst{\sigma}{a}]{\I}{X}
\end{align*}
\end{definition}
\begin{corollary}[Admissible adjoints] \label{cor:adjointUsubst}
\def\Ialta{\iadjointSubst{\sigma}{\Ialt}}%
If \(\iget[state]{\I}=\iget[state]{\It}\) on \(\freevarsdef{\sigma}\),
then \(\iget[const]{\Ia}=\iget[const]{\Ita}\).
If \(\iget[state]{\I}=\iget[state]{\It}\) on $\scomplement{U}$ and $\sigma$ is $U$-admissible for $\theta$ (or $\phi$ or $\alpha$, respectively), then
\begin{align*}
  \ivalues{\Ia}{\theta} &= \ivalues{\Ita}{\theta}
  ~\text{i.e.,}~
  \ivaluation{\Iaz}{\theta} = \ivaluation{\Itaz}{\theta} ~\text{for all states}~\iget[state]{\Iaz} \in \linterpretations{\Sigma}{\allvars}
  \\
  \imodel{\Ia}{\phi} &= \imodel{\Ita}{\phi}\\
  \iaccess[\alpha]{\Ia} &= \iaccess[\alpha]{\Ita}
  ~\text{i.e.,}~
  \iwinreg[\alpha]{\Ia}{X} = \iwinreg[\alpha]{\Ita}{X} ~\text{for all sets}~X \subseteq \linterpretations{\Sigma}{\allvars}
\end{align*}
\end{corollary}
\begin{proofatend}
$\iget[const]{\Ia}$ is well-defined, as $\iget[const]{\Ia}(f)$ is a smooth function since its substitute term \({\applysubst{\sigma}{f}(\usarg)}\) has smooth values.
First, \(\iget[const]{\Ia}(a)(X) = \iwinreg[\applysubst{\sigma}{a}]{\I}{X} = \iget[const]{\Ita}(a)(X)\) holds for all $X\subseteq\linterpretations{\Sigma}{\allvars}\) because the adjoint to $\sigma$ for $\iportray{\I}$ in the case of game symbols is independent of $\iget[state]{\Ia}$ (games have access to the entire state at runtime).
%
By \rref{lem:coincidence-term},
\(\ivaluation{\imodif[const]{\I}{\,\usarg}{d}}{\applysubst{\sigma}{f}(\usarg)}
= \ivaluation{\imodif[const]{\It}{\,\usarg}{d}}{\applysubst{\sigma}{f}(\usarg)}\) when \(\iget[state]{\I}=\iget[state]{\It}\) on \(\freevarsdef{\applysubst{\sigma}{f}(\usarg)} \subseteq \freevarsdef{\sigma}\).
Also \(\imodels{\imodif[const]{\I}{\,\usarg}{d}}{\applysubst{\sigma}{p}(\usarg)}\)
iff \(\imodels{\imodif[const]{\It}{\,\usarg}{d}}{\applysubst{\sigma}{p}(\usarg)}\)
by \rref{lem:coincidence} when \(\iget[state]{\I}=\iget[state]{\It}\) on \(\freevarsdef{\applysubst{\sigma}{p}(\usarg)}\) $\subseteq$ \(\freevarsdef{\sigma}\).
Thus, \(\iget[const]{\Ia}=\iget[const]{\Ita}\) when $\iget[state]{\I}=\iget[state]{\It}$ on $\freevarsdef{\sigma}$.

If $\sigma$ is $U$-admissible for $\phi$ (or $\theta$ or $\alpha$), then
  \(\freevarsdef{\applysubst{\sigma}{f(\usarg)}}\cap U=\emptyset\), so
  \(\scomplement{U}\supseteq\freevarsdef{\applysubst{\sigma}{f(\usarg)}}\)
  for every function symbol $f\in\intsigns{\phi}$ (or $\theta$ or $\alpha$) and likewise for predicate symbols $p\in\intsigns{\phi}$.
  Since \(\iget[state]{\I}=\iget[state]{\It}\) on $\scomplement{U}$ was assumed,
  \(\iget[const]{\Ia}=\iget[const]{\Ita}\) on the function and predicate symbols in $\intsigns{\phi}$ (or $\theta$ or $\alpha$).
  Finally \(\iget[const]{\Ia}=\iget[const]{\Ita}\) on $\intsigns{\phi}$ (or $\intsigns{\theta}$ respectively) implies that
  \(\imodel{\Ita}{\phi}=\imodel{\Ia}{\phi}\) by \rref{lem:coincidence}
  (since \(\imodels{\Itaz}{\phi}\) iff \(\imodels{\Iaz}{\phi}\) holds for all $\iget[state]{\Itaz}$ which trivially satisfy $\iget[state]{\Iaz}=\iget[state]{\Itaz}$ on $\freevarsdef{\phi}$)
  and that \(\ivalues{\Ia}{\theta} = \ivalues{\Ita}{\theta}\) by \rref{lem:coincidence-term}, respectively.
  Similarly, \(\iget[const]{\Ia}=\iget[const]{\Ita}\) on $\intsigns{\alpha}$ implies by \rref{lem:coincidence-HG} that
  \(\iwinreg[\alpha]{\Ia}{X} = \iwinreg[\alpha]{\Ita}{X}\),
  because it implies:
  \(\iwin[\alpha]{\Iaz}{X} = \iwinreg[\alpha]{\Iaz}{\restrictto{X}{\allvars}}\)
  iff
  \(\iwin[\alpha]{\Itaz}{\restrictto{X}{\allvars}} = \iwinreg[\alpha]{\Itaz}{X}\),
  for all $\iget[state]{\Itaz}$ which satisfy $\iget[state]{\Iaz}=\iget[state]{\Itaz}$ on $\allvars\supseteq\freevarsdef{\alpha}$.
  This uses \(\restrictto{X}{\allvars}=X\) from \rref{rem:projections}(\ref{case:restrict-extreme}).
\qedhere
\end{proofatend}

Substituting equals for equals is sound by the compositional semantics of \dL.
The more general uniform substitutions are still sound, because the semantics of uniform substitutes of expressions agrees with the semantics of the expressions themselves in the adjoint interpretations.
The semantic modification of adjoint interpretations has the same effect as the syntactic uniform substitution.

\begin{lemma}[Uniform substitution for terms] \label{lem:usubst-term}
The uniform substitution $\sigma$ and its adjoint interpretation $\iportray{\Ia}$ for $\iportray{\I}$ have the same semantics for all \emph{terms} $\theta$:
\[\ivaluation{\I}{\applyusubst{\sigma}{\theta}} = \ivaluation{\Ia}{\theta}\]
\end{lemma}
\begin{proofatend}
The proof follows from \dL \cite[Lem.\,23]{DBLP:journals/jar/Platzer17}, since the term semantics and the coincidence lemmas for terms that the proof is based on are the same in \dGL. \qedhere
\end{proofatend}

\noindent
The uniform substitute of a formula is true in an interpretation iff the formula itself is true in its adjoint interpretation.
Uniform substitution lemmas are proved by simultaneous induction, since formulas and games are mutually recursive.

\begin{lemma}[Uniform substitution for formulas] \label{lem:usubst}
The uniform substitution $\sigma$ and its adjoint interpretation $\iportray{\Ia}$ for $\iportray{\I}$ have the same semantics for all \emph{formulas} $\phi$:
\[\imodels{\I}{\applyusubst{\sigma}{\phi}} ~\text{iff}~ \imodels{\Ia}{\phi}\]
\end{lemma}
\ifkeepproof
\else
\begin{proof}
The proof is by structural induction on $\phi$ and the structure of $\sigma$, simultaneously with \rref{lem:usubst-HG}.
It is in \rref{app:proofs} with this case for modalities:
\begin{compactenum}
\addtocounter{enumi}{5}
\item
  \(\imodels{\I}{\applyusubst{\sigma}{\usubstgroup{\ddiamond{\alpha}{\phi}}}}\)
  iff \(\imodels{\I}{\ddiamond{\applyusubst{\sigma}{\alpha}}{\applyusubst{\sigma}{\phi}}}\)
  = \(\iwinreg[\applyusubst{\sigma}{\alpha}]{\I}{\imodel{\I}{\applyusubst{\sigma}{\phi}}}\)
  (provided $\sigma$ is $\boundvarsdef{\applyusubst{\sigma}{\alpha}}$-admissible for $\phi$)
  iff (by \rref{lem:bound})
  \(\iwin[\applyusubst{\sigma}{\alpha}]{\I}{\iselectlike{\imodel{\I}{\applyusubst{\sigma}{\phi}}}{\I}{\scomplement{\boundvarsdef{\applyusubst{\sigma}{\alpha}}}}}\).
  
  Starting conversely:
  \(\imodels{\Ia}{\ddiamond{\alpha}{\phi}}\)
  = \(\iwinreg[\alpha]{\Ia}{\imodel{\Ia}{\phi}}\)
  iff (by \rref{lem:usubst-HG})
  \(\iwin[\applyusubst{\sigma}{\alpha}]{\I}{\imodel{\Ia}{\phi}}\)
  iff (by \rref{lem:bound})
  \(\iwin[\applyusubst{\sigma}{\alpha}]{\I}{\iselectlike{\imodel{\Ia}{\phi}}{\I}{\scomplement{\boundvarsdef{\applyusubst{\sigma}{\alpha}}}}}\).
  
  Consequently, it suffices to show that both winning conditions are equal:
  \[{\iselectlike{\imodel{\I}{\applyusubst{\sigma}{\phi}}}{\I}{\scomplement{\boundvarsdef{\applyusubst{\sigma}{\alpha}}}}}
  = {\iselectlike{\imodel{\Ia}{\phi}}{\I}{\scomplement{\boundvarsdef{\applyusubst{\sigma}{\alpha}}}}}\]
  For this, consider any \(\iget[state]{\It}=\iget[state]{\I}\) on $\scomplement{\boundvarsdef{\applyusubst{\sigma}{\alpha}}}$ and show:
  \(\imodels{\It}{\applyusubst{\sigma}{\phi}}\)
  iff \(\imodels{\Iat}{\phi}\).
  By induction hypothesis,
  \(\imodels{\It}{\applyusubst{\sigma}{\phi}}\)
  iff 
  \(\imodels{\Ita}{\phi}\)
  iff
  \(\imodels{\Iat}{\phi}\)
  by \rref{cor:adjointUsubst}, because
  \(\iget[state]{\It}=\iget[state]{\I}\) on $\scomplement{\boundvarsdef{\applyusubst{\sigma}{\alpha}}}$
  and $\sigma$ is $\boundvarsdef{\applyusubst{\sigma}{\alpha}}$-admissible for $\phi$.
  \qedhere
\end{compactenum}
\end{proof}
\fi%
\begin{proofatend}
The proof is by structural induction lexicographically on the structure of $\sigma$ and of $\phi$, with a simultaneous induction in the proof of \rref{lem:usubst-HG}.
\begin{compactenum}
\item
  \(\imodels{\I}{\applyusubst{\sigma}{\usubstgroup{\theta\geq\eta}}}\)
  iff \(\imodels{\I}{\applyusubst{\sigma}{\theta} \geq \applyusubst{\sigma}{\eta}}\)
  iff \(\ivaluation{\I}{\applyusubst{\sigma}{\theta}} \geq \ivaluation{\I}{\applyusubst{\sigma}{\eta}}\),
  by \rref{lem:usubst-term},
  iff \(\ivaluation{\Ia}{\theta} \geq \ivaluation{\Ia}{\eta}\)
  iff \(\imodels{\Ia}{\theta\geq\eta}\).

\item
  \(\imodels{\I}{\applyusubst{\sigma}{\usubstgroup{p(\theta)}}}\)
  iff \(\imodels{\I}{(\applyusubst{\sigma}{p})\big(\applyusubst{\sigma}{\theta}\big)}\)
  iff \(\imodels{\I}{\applyusubst{\{\usarg\mapsto\applyusubst{\sigma}{\theta}\}}{\applysubst{\sigma}{p(\usarg)}}}\)
  iff \(\imodels{\Iminner}{\applysubst{\sigma}{p(\usarg)}}\) by IH as \(\{\usarg\mapsto\applyusubst{\sigma}{\theta}\}\) is simpler than $\sigma$,
  iff \(d \in \iget[const]{\Ia}(p)\)
  iff \((\ivaluation{\Ia}{\theta}) \in \iget[const]{\Ia}(p)\)
  iff \(\imodels{\Ia}{p(\theta)}\)
  with \(d\mdefeq\ivaluation{\I}{\applyusubst{\sigma}{\theta}} = \ivaluation{\Ia}{\theta}\)
  by \rref{lem:usubst-term} for \(\applyusubst{\sigma}{\theta}\).
  The IH for \(\applyusubst{\{\usarg\mapsto\applyusubst{\sigma}{\theta}\}}{\applysubst{\sigma}{p(\usarg)}}\)
  is used on the possibly bigger formula \({\applysubst{\sigma}{p(\usarg)}}\) but the structurally simpler uniform substitution \(\applyusubst{\{\usarg\mapsto\applyusubst{\sigma}{\theta}\}}{}\) that is a mere substitution on function symbol $\usarg$ of arity zero, not a substitution of predicates.

\item
  \(\imodels{\I}{\applyusubst{\sigma}{\usubstgroup{\lnot\phi}}}\)
  iff \(\imodels{\I}{\lnot\applyusubst{\sigma}{\phi}}\)
  iff \(\inonmodels{\I}{\applyusubst{\sigma}{\phi}}\)
  by IH
  iff \(\inonmodels{\Ia}{\phi}\)
  iff \(\imodels{\Ia}{\lnot\phi}\)

\item
  \(\imodels{\I}{\applyusubst{\sigma}{\usubstgroup{\phi\land\psi}}}\)
  iff \(\imodels{\I}{\applyusubst{\sigma}{\phi} \land \applyusubst{\sigma}{\psi}}\)
  iff \(\imodels{\I}{\applyusubst{\sigma}{\phi}}\) and \(\imodels{\I}{\applyusubst{\sigma}{\psi}}\),
  by induction hypothesis,
  iff \(\imodels{\Ia}{\phi}\) and \(\imodels{\Ia}{\psi}\)
  iff \(\imodels{\Ia}{\phi\land\psi}\)

\item
\def\Imd{\imodif[state]{\I}{x}{d}}%
\def\Iamd{\imodif[state]{\Ia}{x}{d}}%
\def\Imda{\iadjointSubst{\sigma}{\Imd}}%
  \(\imodels{\I}{\applyusubst{\sigma}{\usubstgroup{\lexists{x}{\phi}}}}\)
  iff \(\imodels{\I}{\lexists{x}{\applyusubst{\sigma}{\phi}}}\)
  (provided that $\sigma$ is $\{x\}$-admissible for $\phi$)
  iff \(\imodels{\Imd}{\applyusubst{\sigma}{\phi}}\) for some $d$,
  so, by induction hypothesis,
  iff \(\imodels{\Imda}{\phi}\) for some $d$,
  which is equivalent to
  \(\imodels{\Iamd}{\phi}\) by \rref{cor:adjointUsubst} as $\sigma$ is $\{x\}$-admissible for $\phi$ and $\iget[state]{\I}=\iget[state]{\Imd}$ on $\scomplement{\{x\}}$.
  Thus, this is equivalent to
  \(\imodels{\Ia}{\lexists{x}{\phi}}\).

\item
  \(\imodels{\I}{\applyusubst{\sigma}{\usubstgroup{\ddiamond{\alpha}{\phi}}}}\)
  iff \(\imodels{\I}{\ddiamond{\applyusubst{\sigma}{\alpha}}{\applyusubst{\sigma}{\phi}}}\)
  = \(\iwinreg[\applyusubst{\sigma}{\alpha}]{\I}{\imodel{\I}{\applyusubst{\sigma}{\phi}}}\)
  (provided $\sigma$ is $\boundvarsdef{\applyusubst{\sigma}{\alpha}}$-admissible for $\phi$)
  iff (by \rref{lem:bound})
  \(\iwin[\applyusubst{\sigma}{\alpha}]{\I}{\iselectlike{\imodel{\I}{\applyusubst{\sigma}{\phi}}}{\I}{\scomplement{\boundvarsdef{\applyusubst{\sigma}{\alpha}}}}}\).
  
  Starting conversely:
  \(\imodels{\Ia}{\ddiamond{\alpha}{\phi}}\)
  = \(\iwinreg[\alpha]{\Ia}{\imodel{\Ia}{\phi}}\)
  iff (by \rref{lem:usubst-HG})
  \(\iwin[\applyusubst{\sigma}{\alpha}]{\I}{\imodel{\Ia}{\phi}}\)
  iff (by \rref{lem:bound})
  \(\iwin[\applyusubst{\sigma}{\alpha}]{\I}{\iselectlike{\imodel{\Ia}{\phi}}{\I}{\scomplement{\boundvarsdef{\applyusubst{\sigma}{\alpha}}}}}\).
  
  Consequently, it suffices to show that both winning conditions are equal:
  \[{\iselectlike{\imodel{\I}{\applyusubst{\sigma}{\phi}}}{\I}{\scomplement{\boundvarsdef{\applyusubst{\sigma}{\alpha}}}}}
  = {\iselectlike{\imodel{\Ia}{\phi}}{\I}{\scomplement{\boundvarsdef{\applyusubst{\sigma}{\alpha}}}}}\]
  For this, consider any \(\iget[state]{\It}=\iget[state]{\I}\) on $\scomplement{\boundvarsdef{\applyusubst{\sigma}{\alpha}}}$ and show:
  \(\imodels{\It}{\applyusubst{\sigma}{\phi}}\)
  iff \(\imodels{\Iat}{\phi}\).
  By induction hypothesis,
  \(\imodels{\It}{\applyusubst{\sigma}{\phi}}\)
  iff 
  \(\imodels{\Ita}{\phi}\)
  iff
  \(\imodels{\Iat}{\phi}\)
  by \rref{cor:adjointUsubst}, because
  \(\iget[state]{\It}=\iget[state]{\I}\) on $\scomplement{\boundvarsdef{\applyusubst{\sigma}{\alpha}}}$
  and $\sigma$ is $\boundvarsdef{\applyusubst{\sigma}{\alpha}}$-admissible for $\phi$.
\qedhere
\end{compactenum}
\end{proofatend}

\noindent
The uniform substitute of a game can be won into $X$ from state $\iget[state]{\I}$ in an interpretation iff the game itself can be won into $X$ from $\iget[state]{\I}$ in its adjoint interpretation.
The most complicated part of the uniform substitution lemma proofs is the case of repetition $\prepeat{\alpha}$, because it has a least fixpoint semantics.
The proof needs to be set up carefully by transfinite induction (instead of induction along the number of program loop iterations, which is finite for hybrid systems).

\begin{lemma}[Uniform substitution for games] \label{lem:usubst-HG}
The uniform substitution $\sigma$ and its adjoint interpretation $\iportray{\Ia}$ for $\iportray{\I}$ have the same semantics for all \emph{games} $\alpha$:
\[
\iwin[{\applyusubst{\sigma}{\alpha}}]{\I}{X}
~\text{iff}~
\iwin[\alpha]{\Ia}{X}
\]
\end{lemma}
\begin{proof}
The proof is by structural induction on $\alpha$, simultaneously with \rref{lem:usubst}, simultaneously for all $\iget[state]{\I}$ and $X$.
\begin{compactenum}
\item
  \(\iwin[\applyusubst{\sigma}{\usubstgroup{a}}]{\I}{X} = \iwinreg[\applysubst{\sigma}{a}]{\I}{X} = \iget[const]{\Ia}(a)(X) = \iwinreg[a]{\Ia}{X}\)
  for game symbol $a$%

\item 
  \(\iwin[\applyusubst{\sigma}{\usubstgroup{\pumod{x}{\theta}}}]{\I}{X}
  = \iwinreg[\pumod{x}{\applyusubst{\sigma}{\theta}}]{\I}{X}\)
  iff \(X \ni \modif{\iget[state]{\I}}{x}{\ivaluation{\I}{\applyusubst{\sigma}{\theta}}}\)
  = \(\modif{\iget[state]{\I}}{x}{\ivaluation{\Ia}{\theta}}\)
  by using \rref{lem:usubst-term}, which is, thus, equivalent to
  \(\iwin[\pumod{x}{\theta}]{\Ia}{X}\).

\item
\newcommand{\Izeta}{\iconcat[state=\varphi(t)]{\I}}
\def\Izetaa{\iadjointSubst{\sigma}{\Izeta}}%
\newcommand{\Iazeta}{\iconcat[state=\varphi(t)]{\Ia}}
  \(\iwin[\applyusubst{\sigma}{\usubstgroup{\pevolvein{\D{x}=\genDE{x}}{\ivr}}}]{\I}{X}
  = \iwinreg[\pevolvein{\D{x}=\applyusubst{\sigma}{\genDE{x}}}
  {\applyusubst{\sigma}{\ivr}}]{\I}{X}\)
  (provided that $\sigma$ is $\{x,\D{x}\}$-admissible for $\genDE{x},\ivr$)
  iff \(\mexists{\varphi:[0,T]\to\linterpretations{\Sigma}{\allvars}}\)
  with \(\varphi(0)=\iget[state]{\I}\) on $\scomplement{\{\D{x}\}}$, \(\varphi(T)\in X\) and for all $t\geq0$:
  \(\D[s]{\varphi(s)}(t) = \ivaluation{\Izeta}{\applyusubst{\sigma}{\genDE{x}}}
  = \ivaluation{\Izetaa}{\genDE{x}}\) by \rref{lem:usubst-term}
  and
  \(\imodels{\Izeta}{\applyusubst{\sigma}{\ivr}}\),
  which, by \rref{lem:usubst}, holds iff
  \(\imodels{\Izetaa}{\ivr}\).
  
  Conversely,
  \(\iwin[\pevolvein{\D{x}=\genDE{x}}{\ivr}]{\Ia}{X}\)
  iff \(\mexists{\varphi:[0,T]\to\linterpretations{\Sigma}{\allvars}}\)
  with \(\varphi(0)=\iget[state]{\I}\) on $\scomplement{\{\D{x}\}}$ and \(\varphi(T)\in X\) and for all $t\geq0$:
  \(\D[s]{\varphi(s)}(t) = \ivaluation{\Iazeta}{\genDE{x}}\)
  and
  \(\imodels{\Iazeta}{\ivr}\).
  Both sides agree since
  \(\ivalues{\Iazeta}{\genDE{x}}=\ivalues{\Izetaa}{\genDE{x}}\) and
  \(\imodel{\Izetaa}{\ivr}=\imodel{\Iazeta}{\ivr}\)
  by \rref{cor:adjointUsubst}
  as $\sigma$ is $\{x,\D{x}\}$-admissible for $\genDE{x}$ and $\ivr$ and \(\iget[state]{\I}=\iget[state]{\Iazeta}\) on $\scomplement{\boundvarsdef{\pevolvein{\D{x}=\genDE{x}}{\ivr}}}\supseteq\scomplement{\{x,\D{x}\}}$ by \rref{lem:bound}.
  
\item
  \(\iwin[\applyusubst{\sigma}{\usubstgroup{\ptest{\ivr}}}]{\I}{X}
  = \iwinreg[\ptest{\applyusubst{\sigma}{\ivr}}]{\I}{X} = \imodel{\I}{\applyusubst{\sigma}{\ivr}} \cap X\)
  iff, by \rref{lem:usubst}, it is the case that 
  \(\iget[state]{\Ia} \in \imodel{\Ia}{\ivr} \cap X\)
  \(= \iwinreg[\ptest{\ivr}]{\Ia}{X}\).

\item  
  \(\iwin[\applyusubst{\sigma}{\usubstgroup{\pchoice{\alpha}{\beta}}}]{\I}{X}
  = \iwinreg[\pchoice{\applyusubst{\sigma}{\alpha}}{\applyusubst{\sigma}{\beta}}]{\I}{X}\)
  = \(\iwinreg[\applyusubst{\sigma}{\alpha}]{\I}{X} \cup \iwinreg[\applyusubst{\sigma}{\beta}]{\I}{X}\),
  which, by induction hypothesis, is equivalent to
  \(\iwin[\alpha]{\Ia}{X}\) or \(\iwin[\beta]{\Ia}{X}\),
  which is 
  \(\iwin[\alpha]{\Ia}{X} \cup \iwinreg[\beta]{\Ia}{X} = \iwinreg[\pchoice{\alpha}{\beta}]{\Ia}{X}\).
  
\item
  \(\iwin[\applyusubst{\sigma}{\usubstgroup{\alpha;\beta}}]{\I}{X}
  = \iwinreg[\applyusubst{\sigma}{\alpha}; \applyusubst{\sigma}{\beta}]{\I}{X}\)
  = \(\iwinreg[\applyusubst{\sigma}{\alpha}]{\I}{\iwinreg[\applyusubst{\sigma}{\beta}]{\I}{X}}\)
  (provided $\sigma$ is $\boundvarsdef{\applyusubst{\sigma}{\alpha}}$-admissible for $\beta$), which holds iff
  \(\iwin[\applyusubst{\sigma}{\alpha}]{\I}{\iselectlike{\iwinreg[\applyusubst{\sigma}{\beta}]{\I}{X}}{\I}{\scomplement{\boundvarsdef{\applyusubst{\sigma}{\alpha}}}}}\)
  by \rref{lem:bound}.
  
  Starting conversely:
  \(\iwin[\alpha;\beta]{\Ia}{X}
  = \iwinreg[\alpha]{\Ia}{\iwinreg[\beta]{\Ia}{X}}\),
  iff, by IH,\\
  \(\iwin[\applyusubst{\sigma}{\alpha}]{\I}{\iwinreg[\beta]{\Ia}{X}}\)
  iff, by Lem.\,\ref{lem:bound},
  \(\iwin[\applyusubst{\sigma}{\alpha}]{\I}{\iselectlike{\iwinreg[\beta]{\Ia}{X}}{\I}{\scomplement{\boundvarsdef{\applyusubst{\sigma}{\alpha}}}}}\).
  
  Consequently, it suffices to show that both winning conditions are equal:
  \[
  {\iselectlike{\iwinreg[\applyusubst{\sigma}{\beta}]{\I}{X}}{\I}{\scomplement{\boundvarsdef{\applyusubst{\sigma}{\alpha}}}}}
  =
  {\iselectlike{\iwinreg[\beta]{\Ia}{X}}{\I}{\scomplement{\boundvarsdef{\applyusubst{\sigma}{\alpha}}}}}
  \]
  Consider any \(\iget[state]{\It}=\iget[state]{\I}\) on $\scomplement{\boundvarsdef{\applyusubst{\sigma}{\alpha}}}$ to show:
  \(\iwin[\applyusubst{\sigma}{\beta}]{\It}{X}\)
  iff \(\iwin[\beta]{\Iat}{X}\).
  By IH, 
  \(\iwin[\applyusubst{\sigma}{\beta}]{\It}{X}\)
  iff \(\iwin[\beta]{\Ita}{X}\)
  iff \(\iwin[\beta]{\Iat}{X}\)
  by \rref{cor:adjointUsubst}, because \(\iget[state]{\It}=\iget[state]{\I}\) on $\scomplement{\boundvarsdef{\applyusubst{\sigma}{\alpha}}}$ and $\sigma$ is $\boundvarsdef{\applyusubst{\sigma}{\alpha}}$-admissible for $\beta$.

\item
{%
\newcommand{\inflop}[2][]{\tau^{#1}(#2)}%
\newcommand{\oinflop}[2][]{\varrho^{#1}(#2)}%
  The case \(\iwin[\applyusubst{\sigma}{\usubstgroup{\prepeat{\alpha}}}]{\I}{X}
  = \iwinreg[\prepeat{(\applyusubst{\sigma}{\alpha})}]{\I}{X}\)
  (provided $\sigma$ is $\boundvarsdef{\applyusubst{\sigma}{\alpha}}$-admissible for $\alpha$)
  uses an equivalent inflationary fixpoint formulation \cite[Thm.\,3.5]{DBLP:journals/tocl/Platzer15}:%
  \allowdisplaybreaks%
  \begin{align*}
    \inflop[0]{X} &\mdefeq X\\
    \inflop[\kappa+1]{X} &\mdefeq X \cup \iwinreg[\applyusubst{\sigma}{\alpha}]{\I}{\inflop[\kappa]{X}} && \kappa+1~\text{a successor ordinal}\\
    \inflop[\lambda]{X} &\mdefeq \cupfold_{\kappa<\lambda} \inflop[\kappa]{X} && \lambda\neq0~\text{a limit ordinal}
  \intertext{%
  where the union \(\inflop[\infty]{X} = \cupfold_{\kappa<\infty} \inflop[\kappa]{X}\) over all ordinals is \(\iwinreg[\prepeat{(\applyusubst{\sigma}{\alpha})}]{\I}{X}\).
  Define a similar fixpoint formulation for the other side \(\iwinreg[\prepeat{\alpha}]{\Ia}{X} = \oinflop[\infty]{X}\):}%
    \oinflop[0]{X} &\mdefeq X\\
    \oinflop[\kappa+1]{X} &\mdefeq X \cup \iwinreg[\alpha]{\Ia}{\oinflop[\kappa]{X}} && \kappa+1~\text{a successor ordinal}\\
    \oinflop[\lambda]{X} &\mdefeq \cupfold_{\kappa<\lambda} \oinflop[\kappa]{X} && \lambda\neq0~\text{a limit ordinal}
  \end{align*}%
  The equivalence
  \(\iwin[\applyusubst{\sigma}{\usubstgroup{\prepeat{\alpha}}}]{\I}{X} = \inflop[\infty]{X}\)
  iff
  \(\iwin[\prepeat{\alpha}]{\Ia}{X} = \oinflop[\infty]{X}\)
  follows from a proof that:
  \[
  \text{for all}~\kappa~
  \text{and all}~X
  ~\text{and all}~\iget[state]{\It}=\iget[state]{\I} ~\text{on $\scomplement{\boundvarsdef{\applyusubst{\sigma}{\alpha}}}$}:
  ~
  \iget[state]{\It} \in \inflop[\kappa]{X} ~\text{iff}~ \iget[state]{\It} \in \oinflop[\kappa]{X}
  \]
  This is proved by induction on ordinal $\kappa$, which is either 0, a limit ordinal $\lambda\neq0$, or a successor ordinal.
  \begin{enumerate}
  \item[$\kappa=0$:]
  \(\iget[state]{\It} \in \inflop[0]{X} ~\text{iff}~ \iget[state]{\It} \in \oinflop[0]{X}\), because both sets equal $X$.
  
  \item[$\lambda$:]
  \(\iget[state]{\It} \in \inflop[\lambda]{X} 
  = \cupfold_{\kappa<\lambda} \inflop[\kappa]{X}\)
  iff there is a $\kappa<\lambda$ such that
  \(\iget[state]{\It} \in \inflop[\kappa]{X}\)
  iff, by IH,
  \(\iget[state]{\It} \in \oinflop[\kappa]{X}\) for some $\kappa<\lambda$,
  iff
  \(\iget[state]{\It} \in \cupfold_{\kappa<\lambda} \oinflop[\kappa]{X}
  = \oinflop[\lambda]{X}\).
  
  \item[$\kappa+1$:]
  \(\iget[state]{\It} \in \inflop[\kappa+1]{X}
  = X \cup \iwinreg[\applyusubst{\sigma}{\alpha}]{\I}{\inflop[\kappa]{X}}\),
  which, by \rref{lem:bound}, is equivalent to
  \(\iget[state]{\It} \in X \cup \iwinreg[\applyusubst{\sigma}{\alpha}]{\I}{\iselectlike{\inflop[\kappa]{X}}{\It}{\scomplement{\boundvarsdef{\applyusubst{\sigma}{\alpha}}}}}\)
  
  Starting from the other end,
  \(\iget[state]{\It} \in \oinflop[\kappa+1]{X}
  = X \cup \iwinreg[\alpha]{\Ia}{\oinflop[\kappa]{X}}\)
  iff, by \rref{cor:adjointUsubst} using \(\iget[state]{\It}=\iget[state]{\I}\) on $\scomplement{\boundvarsdef{\applyusubst{\sigma}{\alpha}}} \supseteq \scomplement{\boundvarsdef{\alpha}}$,
  \(\iget[state]{\It} \in X \cup \iwinreg[\alpha]{\Ita}{\oinflop[\kappa]{X}}\)
  iff, by induction hypothesis on $\alpha$,
  \(\iget[state]{\It} \in X \cup \iwinreg[\applyusubst{\sigma}{\alpha}]{\It}{\oinflop[\kappa]{X}}\)
  iff, by \rref{lem:bound},
  \(\iget[state]{\It} \in X \cup \iwinreg[\applyusubst{\sigma}{\alpha}]{\It}{\iselectlike{\oinflop[\kappa]{X}}{\It}{\scomplement{\boundvarsdef{\applyusubst{\sigma}{\alpha}}}}}\)
  Consequently, it suffices to show that both winning conditions are equal:
  \(
  {\iselectlike{\inflop[\kappa]{X}}{\It}{\scomplement{\boundvarsdef{\applyusubst{\sigma}{\alpha}}}}}
  =
  {\iselectlike{\oinflop[\kappa]{X}}{\It}{\scomplement{\boundvarsdef{\applyusubst{\sigma}{\alpha}}}}}
  \).
  Consider any state \(\iget[state]{\Iz}=\iget[state]{\I}\) on $\scomplement{\boundvarsdef{\applyusubst{\sigma}{\alpha}}}$,
  then
  \(\iget[state]{\Iz} \in \inflop[\kappa]{X}\)
  iff
  \(\iget[state]{\Iz} \in \oinflop[\kappa]{X}\)
  by induction hypothesis on $\kappa<\kappa+1$.
  \end{enumerate}
}

\item
  \(\iwin[\applyusubst{\sigma}{\usubstgroup{\pdual{\alpha}}}]{\I}{X}
  = \iwinreg[\pdual{(\applyusubst{\sigma}{\alpha})}]{\I}{X}
  = \scomplement{\big(\iwinreg[\applyusubst{\sigma}{\alpha}]{\I}{\scomplement{X}}\big)}\)
  iff \(\inowin[\applyusubst{\sigma}{\alpha}]{\I}{\scomplement{X}}\),
  which, by IH, is equivalent to \(\inowin[\alpha]{\Ia}{\scomplement{X}}\),
  which is, in turn, equivalent to
  \(\iget[state]{\Ia} \in \scomplement{\big(\iwinreg[\alpha]{\Ia}{\scomplement{X}}\big)}
  = \iwinreg[\pdual{\alpha}]{\Ia}{X}\).
  \qedhere

\end{compactenum}
\end{proof}

\subsection{Soundness}

Soundness of uniform substitution for \dGL now follows from the above uniform substitution lemmas with the same proof that it had from corresponding lemmas in \dL \cite{DBLP:journals/jar/Platzer17}.
Due to the modular setup of uniform substitutions, the change from \dL to \dGL is reflected in how the uniform substitution lemmas are proved, not in how they are used for the soundness of proof rule \irref{US}.
A proof rule is \emph{sound} iff validity of all its premises implies validity of its conclusion.

\begin{theorem}[Soundness of uniform substitution]%
  \label{thm:usubst-sound}%
  Proof rule \irref{US} is sound.%
  {\upshape\[
  \cinferenceRuleQuote{US}
  \]}
\end{theorem}
\begin{proofatend}
\def\Ia{\iadjointSubst{\sigma}{\I}}%
Let the premise $\phi$ of \irref{US} be valid, i.e., \m{\imodels{\I}{\phi}} for all interpretations $\iget[const]{\I}$ and states $\iget[state]{\I}$.
To show that the conclusion is valid, consider any interpretation $\iget[const]{\I}$ and state $\iget[state]{\I}$ and show \(\imodels{\I}{\applyusubst{\sigma}{\phi}}\).
By \rref{lem:usubst}, \(\imodels{\I}{\applyusubst{\sigma}{\phi}}\) iff \(\imodels{\Ia}{\phi}\).
Now \(\imodels{\Ia}{\phi}\) holds, because \(\imodels{\I}{\phi}\) for all $\iportray{\I}$, including $\iportray{\Ia}$, by premise.
\qedhere
\end{proofatend}

\noindent
As in \dL, uniform substitutions can soundly instantiate locally sound proof rules or proofs \cite{DBLP:journals/jar/Platzer17} just like proof rule \irref{US} soundly instantiates axioms or other valid formulas (\rref{thm:usubst-sound}).
An inference or proof rule is \emph{locally sound} iff its conclusion is valid in any interpretation $\iget[const]{\I}$ in which all its premises are valid.
All locally sound proof rules are sound.
The use of \rref{thm:usubst-rule} in a proof is marked \irref{USR}.%

\begin{theorem}[Soundness of uniform substitution of rules] \label{thm:usubst-rule}
  If $\freevarsdef{\sigma}=\emptyset$, all uniform substitution instances of locally sound inferences are locally sound:
  \[
\linfer
{\phi_1 \quad \dots \quad \phi_n}
{\psi}
~~\text{locally sound}\qquad\text{implies}\qquad
\linfer%
{\applyusubst{\sigma}{\phi_1} \quad \dots \quad \applyusubst{\sigma}{\phi_n}}
{\applyusubst{\sigma}{\psi}}
~~\text{locally sound}
\irlabel{USR|USR}
  \]
\end{theorem}
\begin{proofatend}
\def\locproof{\mathcal{D}}%
Let $\locproof$ be the inference on the left and $\applyusubst{\sigma}{\locproof}$ the substituted inference on the right.
Assume $\locproof$ to be locally sound.
To show that $\applyusubst{\sigma}{\locproof}$ is locally sound, consider any $\iget[const]{\I}$ in which all premises of $\applyusubst{\sigma}{\locproof}$ are valid, i.e.,
\(\iget[const]{\I}\models{\applyusubst{\sigma}{\phi_j}}\) for all $j$, i.e.,
\(\imodels{\I}{\applyusubst{\sigma}{\phi_j}}\) for all $\iget[state]{\I}$ and all $j$.
By \rref{lem:usubst}, \(\imodels{\I}{\applyusubst{\sigma}{\phi_j}}\) is equivalent to
\(\imodels{\Ia}{\phi_j}\),
which, thus, also holds for all $\iget[state]{\I}$ and all $j$.
By \rref{cor:adjointUsubst}, \(\imodel{\Ia}{\phi_j}=\imodel{\Ita}{\phi_j}\) for all $\iget[state]{\Ita}$, since $\freevarsdef{\sigma}=\emptyset$.
Fix an arbitrary state $\iget[state]{\Ita}$.
Then \(\imodels{\Itar}{\applyusubst{\sigma}{\phi_j}}\) holds for all $\iget[state]{\Itar}$ and all $j$ for the same (arbitrary) $\iget[state]{\Ita}$ that determines $\iget[const]{\Ita}$.

Consequently, all premises of $\locproof$ are valid in the same $\iget[const]{\Ita}$, i.e. \(\iget[const]{\Ita}\models{\phi_j}\) for all $j$.
Thus, \(\iget[const]{\Ita}\models{\psi}\) by local soundness of $\locproof$.
That is, \(\imodels{\Ia}{\psi}=\imodel{\Ita}{\psi}\) by \rref{cor:adjointUsubst} for all $\iget[state]{\Ia}$.
By \rref{lem:usubst}, \(\imodels{\Ia}{\psi}\) is equivalent to \(\imodels{\I}{\applyusubst{\sigma}{\psi}}\),
which continues to hold for all $\iget[state]{\I}$.
Thus, \(\iget[const]{\I}\models{\applyusubst{\sigma}{\psi}}\), i.e., the conclusion of $\applyusubst{\sigma}{\locproof}$ is valid in $\iget[const]{\I}$, hence $\applyusubst{\sigma}{\locproof}$ is locally sound.
Consequently, all uniform substitution instances $\applyusubst{\sigma}{\locproof}$ of locally sound inferences $\locproof$ with $\freevarsdef{\sigma}=\emptyset$ are locally sound.
\qedhere
\end{proofatend}

\section{Axioms}

Axioms and axiomatic proof rules for differential game logic are listed in \rref{fig:dGL}, where $\usall$ is the (finite-dimensional) vector of all relevant variables.
The axioms are concrete \dGL formulas that are valid.
The axiomatic proof rules are concrete formulas for the premises and concrete formulas for the conclusion that are locally sound.
This makes \rref{fig:dGL} straightforward to implement by copy-and-paste.
\rref{thm:usubst-sound} can be used to instantiate axioms to other \dGL formulas.
\rref{thm:usubst-rule} can be used to instantiate axiomatic proof rules to other concrete \dGL inferences.
Complete axioms for first-order logic from elsewhere \cite{DBLP:journals/jar/Platzer17} and a proof rule (written \irref{qear}\irlabel{qear|\usebox{\Rval}}) for decidable real arithmetic \cite{tarski_decisionalgebra51} are assumed as a basis.

\begin{figure}[tb]
  \centering
  \renewcommand*{\irrulename}[1]{\text{#1}}%
  \renewcommand{\linferenceRuleNameSeparation}{~~}
  \newdimen\linferenceRulehskipamount%
  \linferenceRulehskipamount=1mm%
  \newdimen\lcalculuscollectionvskipamount%
  \lcalculuscollectionvskipamount=0.1em%
  \begin{calculuscollections}{\columnwidth}
    \begin{calculus}
      \cinferenceRule[box|$\dibox{\cdot}$]{box axiom}
      {\linferenceRule[equiv]
        {\lnot\ddiamond{a}{\lnot p(\usall)}}
        {\axkey{\dbox{a}{p(\usall)}}}
      }
      {}
      \cinferenceRule[assignd|$\didia{:=}$]{assignment / substitution axiom}
      {\linferenceRule[equiv]
        {p(f)}
        {\axkey{\ddiamond{\pupdate{\umod{x}{f}}}{p(x)}}}
      }
      {}%
      \cinferenceRule[DSd|DS]{(constant) differential equation solution} %
      {\linferenceRule[viuqe]
        {\axkey{\ddiamond{\pevolve{\D{x}=f}}{p(x)}}}
        {\lexists{t{\geq}0}{\ddiamond{\pupdate{\pumod{x}{x+f\itimes t}}}{p(x)}}}
      }
      {}
      \cinferenceRule[testd|$\didia{?}$]{test}
      {\linferenceRule[equiv]
        {q \land p}
        {\axkey{\ddiamond{\ptest{q}}{p}}}
      }{}
      \cinferenceRule[choiced|$\didia{\cup}$]{axiom of nondeterministic choice}
      {\linferenceRule[equiv]
        {\ddiamond{a}{p(\usall)} \lor \ddiamond{b}{p(\usall)}}
        {\axkey{\ddiamond{\pchoice{a}{b}}{p(\usall)}}}
      }{}
      \cinferenceRule[composed|$\didia{{;}}$]{composition}
      {\linferenceRule[equiv]
        {\ddiamond{a}{\ddiamond{b}{p(\usall)}}}
        {\axkey{\ddiamond{a;b}{p(\usall)}}}
      }{}
      \cinferenceRule[iterated|$\didia{{}^*}$]{iteration/repeat unwind pre-fixpoint, even fixpoint}
      {\linferenceRule[equiv]
        {p(\usall) \lor \ddiamond{a}{\ddiamond{\prepeat{a}}{p(\usall)}}}
        {\axkey{\ddiamond{\prepeat{a}}{p(\usall)}}}
      }{}%
      \cinferenceRule[duald|$\didia{{^d}}$]{dual}
      {\linferenceRule[equiv]
        {\lnot\ddiamond{a}{\lnot p(\usall)}}
        {\axkey{\ddiamond{\pdual{a}}{p(\usall)}}}
      }{}
    \end{calculus}
    \qquad
    \begin{calculus}
      \cinferenceRule[M|M]{$\ddiamond{}{}$ monotone / $\ddiamond{}{}$-generalization} %
      {\linferenceRule[formula]
        {p(\usall)\limply q(\usall)}
        {\ddiamond{a}{p(\usall)}\limply\ddiamond{a}{q(\usall)}}
      }{}
      \cinferenceRule[FP|FP]{iteration is least fixpoint / reflexive transitive closure RTC, equivalent to invind in the presence of R}
      {\linferenceRule[formula]
        {p(\usall) \lor \ddiamond{a}{q(\usall)} \limply q(\usall)}
        {\ddiamond{\prepeat{a}}{p(\usall)} \limply q(\usall)}
      }{}
      \cinferenceRule[MP|MP]{modus ponens}
      {\linferenceRule[formula]
        {p \quad p\limply q}
        {q}
      }{}%
      \cinferenceRule[gena|$\forall$]{$\forall{}$ generalisation}
      {\linferenceRule[formula]
        {p(x)}
        {\lforall{x}{p(x)}}
      }{}%
    \end{calculus}%
  \end{calculuscollections}
  \vspace*{-\baselineskip} %
  \caption{Differential game logic axioms and axiomatic proof rules}
  \label{fig:dGL}
\end{figure}%

\newcommand{\precond}{j(x)}%
The axiom \irref{composed}, for example, expresses that Angel has a winning strategy in game $a;b$ to achieve $p(\usall)$ if and only if she has a winning strategy in game $a$ to achieve \(\ddiamond{b}{p(\usall)}\), i.e., to reach the region from which she has a winning strategy in game $b$ to achieve $p(\usall)$.
Rule \irref{US} can instantiate axiom \irref{composed}, for example, with
\(\sigma = \usubstlist{\usubstmod{a}{\pdual{(\pchoice{\pupdate{\pumod{v}{2}}}{\pupdate{\pumod{v}{\sndvel}}})}},
\usubstmod{b}{\pevolve{\D{x}=v}},
\usubstmod{p(\usall)}{x>0}}\)
to prove
\[
\ddiamond{\pdual{(\pchoice{\pupdate{\pumod{v}{2}}}{\pupdate{\pumod{v}{\sndvel}}})}; \pevolve{\D{x}=v}}{x>0}
\lbisubjunct
\ddiamond{\pdual{(\pchoice{\pupdate{\pumod{v}{2}}}{\pupdate{\pumod{v}{\sndvel}}})}} {\ddiamond{\pevolve{\D{x}=v}}{x>0}}
\]
The right-hand formula can be simplified when using \irref{US} again to instantiate axiom \irref{duald} with
\(\sigma = \usubstlist{\usubstmod{a}{\pchoice{\pupdate{\pumod{v}{2}}}{\pupdate{\pumod{v}{\sndvel}}}},
\usubstmod{p(\usall)}{\ddiamond{\pevolve{\D{x}=v}}{x>0}}}\)
to prove
\[
\ddiamond{\pdual{(\pchoice{\pupdate{\pumod{v}{2}}}{\pupdate{\pumod{v}{\sndvel}}})}} {\ddiamond{\pevolve{\D{x}=v}}{x>0}}
\lbisubjunct
\lnot\ddiamond{\pchoice{\pupdate{\pumod{v}{2}}}{\pupdate{\pumod{v}{\sndvel}}}}{\lnot\ddiamond{\pevolve{\D{x}=v}}{x>0}}
\]
When eliding the equivalences and writing down the resulting formula along with the axiom that was uniformly substituted to obtain it, this yields a proof:
\begin{sequentdeduction}[array]
\linfer[composed]
{\linfer[duald]
  {\linfer[choiced]
    {\linfer[assignd]
      {\linfer[DSd]
        {\linfer[assignd]
            {\lsequent{\precond} {\lnot(\lnot\lexists{t{\geq}0}{x+2t>0} \lor \lnot\lexists{t{\geq}0}{x+(\sndvel)t>0})}}
          {\lsequent{\precond} {\lnot(\lnot\lexists{t{\geq}0}{\ddiamond{\pupdate{\pumod{x}{x+2\itimes t}}}{x>0}} \lor \ddiamond{\pupdate{\pumod{v}{\sndvel}}}{\lnot\lexists{t{\geq}0}{\ddiamond{\pupdate{\pumod{x}}{x+v\itimes t}}{x>0}}})}}
        }%
        {\lsequent{\precond} {\lnot(\lnot\ddiamond{\pevolve{\D{x}=2}}{x>0} \lor \ddiamond{\pupdate{\pumod{v}{\sndvel}}}{\lnot\ddiamond{\pevolve{\D{x}=v}}{x>0}})}}
      }%
      {\lsequent{\precond} {\lnot(\ddiamond{\pupdate{\pumod{v}{2}}}{\lnot\ddiamond{\pevolve{\D{x}=v}}{x>0}} \lor \ddiamond{\pupdate{\pumod{v}{\sndvel}}}{\lnot\ddiamond{\pevolve{\D{x}=v}}{x>0}})}}
    }%
    {\lsequent{\precond} {\lnot\ddiamond{\pchoice{\pupdate{\pumod{v}{2}}}{\pupdate{\pumod{v}{\sndvel}}}}{\lnot\ddiamond{\pevolve{\D{x}=v}}{x>0}}}}
  }%
{\lsequent{\precond} {\ddiamond{\pdual{(\pchoice{\pupdate{\pumod{v}{2}}}{\pupdate{\pumod{v}{\sndvel}}})}} {\ddiamond{\pevolve{\D{x}=v}}{x>0}}}}
}%
{\lsequent{\precond} {\ddiamond{\pdual{(\pchoice{\pupdate{\pumod{v}{2}}}{\pupdate{\pumod{v}{\sndvel}}})}; \pevolve{\D{x}=v}}{x>0}
}}
\end{sequentdeduction}
It is soundness-critical that \irref{US} checks velocity $v$ is not bound in the ODE when substituting it for $f$ in \irref{DSd}, since \(x+v\itimes t\) is not, otherwise, the correct solution of \(\D{x}=v\).
Likewise, the velocity assignment \(\pupdate{\pumod{v}{\sndvel}}\) cannot soundly be substituted into the differential equation via \irref{assignd}, which \irref{US} prevents as $x$ is bound in \m{\D{x}=v}.
Instead, axiom \irref{assignd} for \(\pupdate{\pumod{v}{\sndvel}}\) needs to be delayed until after solving by \irref{DSd}.
If it were \m{\pupdate{\pumod{v}{x^2{+}1}}} instead of \m{\pupdate{\pumod{v}{\sndvel}}}, then rule \irref{qear} would finish the proof.
But for the above proof with \m{\pupdate{\pumod{v}{\sndvel}}} to finish, extra assumptions need to be identified.

With
\(\sigma = \usubstlist{\usubstmod{a}{\pdual{(\pchoice{\pupdate{\pumod{v}{2}}}{\pupdate{\pumod{v}{\sndvel}}})}; \pevolve{\D{x}=v}},
\usubstmod{p(\usall)}{x{>}0},
\usubstmod{q(\usall)}{x^2{>}0}}\),
\irref{USR} instantiates axiomatic rule \irref{M} to prove an inference continuing the proof:
\[
\hspace{0.5cm}
\linfer[USR+M]
{{x{>}0} \limply {x^2{>}0}}
{{\ddiamond{\pdual{(\pchoice{\pupdate{\pumod{v}{2}}}{\pupdate{\pumod{v}{\sndvel}}})}; \pevolve{\D{x}=v}}{x{>}0}
}
\limply {\ddiamond{\pdual{(\pchoice{\pupdate{\pumod{v}{2}}}{\pupdate{\pumod{v}{\sndvel}}})}; \pevolve{\D{x}=v}}{x^2{>}0}
}}
\]
Variable $x$ can be used in the postconditions despite being bound in the game.
Likewise, rule \irref{USR} can instantiate the above proof with
\(\sigma = \usubstlist{\usubstmod{j(\usarg)}{\usarg{>}{-}1}}\) to:%
\vspace{-0.5\baselineskip}%
\begin{sequentdeduction}[array]
\linfer[USR]
{\linfer[qear]
  {\lclose}
  {\lsequent{x>-1} {\lnot(\lnot\lexists{t{\geq}0}{x+2t>0} \lor \lnot\lexists{t{\geq}0}{x+(\sndvel)t>0})}}
}%
{\lsequent{x>-1} 
{\ddiamond{\pdual{(\pchoice{\pupdate{\pumod{v}{2}}}{\pupdate{\pumod{v}{\sndvel}}})}; \pevolve{\D{x}=v}}{x>0}
}}
\end{sequentdeduction}
\irref{USR} soundly instantiates the inference from premise to conclusion of the proof without having to change or repeat any part of the proof.
Uniform substitutions enable flexible but sound reasoning forwards, backwards, on proofs, or mixed \cite{DBLP:journals/jar/Platzer17}.
Without \irref{USR}, these features would complicate soundness-critical prover cores.

Since the axioms and axiomatic proof rules in \rref{fig:dGL} are themselves instances of axiom schemata and proof rule schemata that axiomatize \dGL \cite{DBLP:journals/tocl/Platzer15}, they are (even locally!) \emph{sound}.
Axiom \irref{DSd} stems from \dL \cite{DBLP:journals/jar/Platzer17} and is for solving constant differential equations.
Now that differentials are available, all differential axioms such as the Leibniz axiom
\(\der{f(\usall)\cdot g(\usall)} = \der{f(\usall)}\cdot g(\usall) + f(\usall)\cdot\der{g(\usall)}\)
and all other axioms for differential equations \cite{DBLP:journals/jar/Platzer17} can be added to \dGL.
Furthermore, hybrid games make it possible to equivalently replace differential equations with evolution domains by hybrid games without domain constraints \cite[Lem.\,3.4]{DBLP:journals/tocl/Platzer15}.

\newcommand{\reduct}[1]{#1^\flat}%
\newcommand{\LBase}{\textit{L}\xspace}%

The converse challenge for \emph{completeness} is to prove that uniform substitutions are flexible enough to prove all required instances of \dGL axioms and axiomatic proof rules.
%
  A \dGL formula $\phi$ is called \emph{surjective} iff rule \irref{US} can instantiate $\phi$ to any of its axiom schema instances, which are those formulas that are obtained by just replacing game symbols $a$ uniformly by any hybrid game etc. 
  An axiomatic rule is called \emph{surjective} iff \irref{USR} can instantiate it to any of its proof rule schema instances.
The axiom \irref{testd} is surjective, as it does not have any bound variables, so its instances are admissible.
Similarly rules \irref{MP} and rule \irref{gena} become surjective \cite{DBLP:journals/jar/Platzer17}.
The proof of the following lemma transfers from prior work \cite[Lem.\,39]{DBLP:journals/jar/Platzer17}, since any hybrid game can be substituted for a game symbol.%
\begin{lemma}[Surjective axioms] \label{lem:surjectiveaxiom}
  If $\phi$ is a \dGL formula that is built only from 
  %
  game symbols but no function or predicate symbols,
  then $\phi$ is surjective.
  Axiomatic rules consisting of surjective \dGL formulas are surjective.
\end{lemma}

\noindent
%

\ifpredicationals
\else
Unfortunately, none of the axioms from \rref{fig:dGL} satisfy the assumptions of \rref{lem:surjectiveaxiom}.
While the argument from previous work would succeed \cite{DBLP:journals/jar/Platzer17}, the trick to simplify the proof is to consider $p(\usall)$ to be \(\ddiamond{c}{\ltrue}\) for some game symbol $c$.
Then any formula $\varphi$ can be instantiated for $p(\usall)$ alias \(\ddiamond{c}{\ltrue}\) by substituting the game symbol $c$ with the game \(\ptest{\varphi}\) and subsequently using the surjective axiom \irref{testd} to replace the resulting \(\ddiamond{\ptest{\varphi}}{\ltrue}\) by \(\varphi\land\ltrue\) or its equivalent \(\varphi\) as intended.
\fi
This makes axioms \irref{box+testd+choiced+composed+iterated+duald} and all axiomatic rules in \rref{fig:dGL} surjective.

With \rref{lem:surjectiveaxiom} to show that all schema instantiations required for completeness are provable by \irref{US+USR} from axioms or axiomatic rules, relative completeness of \dGL follows immediately from a previous schematic completeness result for \dGL \cite{DBLP:journals/tocl/Platzer15} and relative completeness of uniform substitution for \dL \cite{DBLP:journals/jar/Platzer17}.

\begin{theorem}[Relative completeness] \label{thm:dGL-complete}%
  The \dGL calculus is a \emph{sound and complete axiomatization} of hybrid games relative to \emph{any} differentially expressive logic\footnote{%
A logic \LBase closed under first-order connectives
is \emph{differentially expressive} (for \dGL) if every \dGL formula $\phi$ has an equivalent $\reduct{\phi}$ in \LBase and all differential equation equivalences of the form \(\ddiamond{\pevolve{\D{x}=\genDE{x}}}{G} \lbisubjunct \reduct{(\ddiamond{\pevolve{\D{x}=\genDE{x}}}{G})}\) for $G$ in \LBase are provable in its calculus.
} \LBase, i.e.,
  every valid \dGL formula is provable in \dGL from \LBase tautologies.
\end{theorem}
\begin{proofatend}
\let\Oracle\LBase%
\newcommand{\precondf}{F}%
\newcommand{\postcondf}{G}%
The axioms and axiomatic proof rules in \rref{fig:dGL} are concrete instances of sound schemata or rules from prior work \cite{DBLP:journals/tocl/Platzer15,DBLP:journals/jar/Platzer17}.
By \rref{lem:surjectiveaxiom} the axioms \irref{box+testd+choiced+composed+iterated+duald} and all axiomatic rules in \rref{fig:dGL} are surjective, so can be instantiated by rule \irref{US} to any of their schema instances.
Except for assignments, these cover all axioms and proof rules used in the relative completeness theorem for \dGL's schematic axiomatization \cite[Thm.\,4.5]{DBLP:journals/tocl/Platzer15}.
Thus, \rref{lem:surjectiveaxiom} makes the previous completeness proof transfer to the axiomatic proof calculus of differential-form \dGL,
but only if all uses of the assignment axiom, which is not surjective, can be patched.
The only such case is in the proof that
\m{\entails \precondf \limply \ddiamond{\pupdate{\pumod{x}{\theta}}}{\postcondf}}
implies that this formula can be proved in the \dGL calculus from \Oracle.
Since \(\ddiamond{\pupdate{\pumod{x}{\theta}}}{\postcondf}\)
is equivalent to
\(\dbox{\pupdate{\pumod{x}{\theta}}}{\postcondf}\) via axiom \irref{box},
this follows from the corresponding case in the completeness proof for \dL \cite[Thm.\,40]{DBLP:journals/jar/Platzer17} that
\m{\entails \precondf \limply \dbox{\pupdate{\pumod{x}{\theta}}}{\postcondf}}
implies that this formula is provable by rule \irref{US} from the \irref{box} dual of assignment axiom \irref{assignd}.
\qedhere
\end{proofatend}

\section{Related Work}

Since the primary impact of uniform substitution is on conceptual simplicity and a significantly simpler prover implementation, this related work discussion focuses on hybrid games theorem proving.
A broader discussion of both hybrid games and uniform substitution themselves is provided in the literature \cite{DBLP:journals/tocl/Platzer15,DBLP:journals/jar/Platzer17}.
The approach presented here also helps discrete game logic \cite{DBLP:conf/focs/Parikh83}, but that is only challenging after a suitable generalization beyond the propositional case.

Prior approaches to hybrid games theorem proving are either based on differential game logic \cite{DBLP:journals/tocl/Platzer15,DBLP:journals/tocl/Platzer17} or on an exterior game embedding of differential dynamic logic \cite{DBLP:conf/cade/QueselP12}.
This paper is based on prior findings on differential game logic \cite{DBLP:journals/tocl/Platzer15} that it complements by giving an \emph{explicit construction} for uniform substitution.
This enables a purely axiomatic version of \dGL that does not need the axiom schemata or proof rule schemata from previous approaches \cite{DBLP:journals/tocl/Platzer15,DBLP:journals/tocl/Platzer17}.
This change makes it substantially simpler to implement \dGL soundly in a theorem prover.
The exterior game embedding of differential dynamic logic \cite{DBLP:conf/cade/QueselP12} was implemented with proof rule schemata in \KeYmaera and was, thus, significantly more complex.

The primary and significant challenge of this paper compared to previous uniform substitution approaches \cite{Church_1956,DBLP:journals/jar/Platzer17,DBLP:conf/cpp/BohrerRVVP17} arose from the semantics of hybrid games, which need a significantly different set-valued winning region style.
The root-cause is that, unlike the normal modal logic \dL, \dGL is a subregular modal logic \cite{DBLP:journals/tocl/Platzer15}.
Especially, Kripke's axiom \(\dbox{\alpha}{(\phi\limply\psi)} \limply (\dbox{\alpha}{\phi} \limply \dbox{\alpha}{\psi})\) is unsound for \dGL.

\section{Conclusion and Future Work}

This paper provides an explicit construction of uniform substitutions and proves it sound for differential game logic.
It also indicates that uniform substitutions are flexible when a logic is changed.
The modularity principles of uniform substitution hold what they promise, making an implementation in a theorem prover exceedingly straightforward.
The biggest challenge was the semantic generalization of the soundness proofs to the subtle interactions caused by hybrid games.

In future work it could be interesting to devise a framework for the general construction of uniform substitutions for arbitrary logics from a certain family.
The challenge is that such an approach partially goes against the spirit of uniform substitution, which is built for flexibility (straightforward and easy to change), not necessarily generality (already preequipped to reconfigure for all possible future changes).
Such generality seems to require a schematic understanding, possibly self-defeating for the simplicity advantages of uniform substitutions.

\section*{Acknowledgment}
I thank Brandon Bohrer and Yong Kiam Tan for their helpful feedback.
This material is based upon work supported by the National Science Foundation under NSF CAREER Award CNS-1054246.

Any opinions, findings, and conclusions or recommendations expressed in this publication are those of the author(s) and do not necessarily reflect the views of the National Science Foundation.

\bibliographystyle{plainurl}
\bibliography{dGL-usubst}

\ifkeepproof
\else
\iflongversion
\clearpage
\appendix
\section{Proofs} \label{app:proofs}
\printproofs
\fi
\fi
\end{document}